\documentclass{acm_proc_article-sp}
 \usepackage[utf8]{inputenc}

\usepackage{multicol}
\usepackage{dsfont}
\usepackage{picins}
\usepackage{amsmath}
\usepackage{amssymb}
\usepackage{latexsym} 
\usepackage{mathrsfs}
\usepackage{amsmath}
\usepackage{amssymb}
\usepackage{enumerate}
\usepackage[all]{xypic}
\usepackage{graphicx}
\newdef{definition}{Definition}
\newdef{lemma}{Lemma}
\newdef{remark}{Remark}
\newdef{theorem}{Theorem}
\newdef{proposition}{Proposition}
\newdef{claim}{Claim}
\usepackage{parallel}

\newcommand{\cutout}[1]{}

\newcommand{\longversion}[1]{}

\newcommand{\mat}[1]{\mathcal{#1}}

\newcommand{\rTo}{\longrightarrow}

\newcommand{\set}[1]{\{\,#1\,\}}

 \newcommand{\tand}{\text{ and }}

 \newcommand{\tst}{\text{ s.t. }}

\newcommand{\myequiv}{\mathfrak{R}}
\newcommand{\actset}{\mat{L}}
\newcommand{\inter}{\bigtriangleup}
\newcommand{\ou}{+}
\newcommand{\delay}[1]{\Theta^{#1}}
\newcommand{\fleche}{\longrightarrow}
\newcommand{\flecheop}{\leadsto}
\newcommand{\reduc}{\hookrightarrow}
\newcommand{\flechep}{\rightarrow}
\newcommand{\dr}[1]{d(#1)}
\newcommand{\evn}[1]{\psi({#1})}

\newcommand{\evns}{\mat{E}}
\newcommand{\clk}{Clk}
\newcommand{\fin}[2]{\mathcal{F}^{#1}(#2)}
\newcommand{\fint}[3]{(\mathcal{F}^{#1}{+ {#2})}(#3)}
\newcommand{\tcsbis}{\sim_{\mathbb{T}}}

\newcommand{\tconf}{\mathbb{C}_{\tau}}
\newcommand{\myet}{\wedge}
\newcommand{\myou}{\vee}
\newcommand{\myetb}{\bigwedge}
\newcommand{\myoub}{\bigvee}
\newcommand{\reff}[3]{\rho({#1},{#2},{#3})}
\newcommand{\st}{\mathbb{C}} 
\newcommand{\sem}[1]{[\![#1]\!]}

\newlength{\myhidel}
\newcommand{\zero}{\mathds{O}}
\newcommand{\aconf}[2]{\langle #1, #2 \rangle}
\newcommand{\for}[1]{\mathbb{F}(#1)}
\newcommand{\as}{\nu}
\newcommand{\dom}{\mathbb{R}^{+}}

\newcommand{\lbul}[1]{{}^{\bullet}{#1}}

\newcounter{casecounter}
\def\thecasecounter{(\roman{casecounter})}
\newenvironment{proofbycases}[1][0]{
  \setcounter{casecounter}{#1}%
}{}

\newenvironment{caseinproof}{%
  \refstepcounter{casecounter}%
  \vskip 1pt
  \noindent%
\textbf{\emph{Case~\thecasecounter.}}~
}{\par}

\newcommand{\myemph}[1]{\textbf{#1}}
\renewcommand{\cong}{\equiv}

\newgraphescape{p}[1]{[]*+{#1}*\cir{}="P"([d(0.2)r(0.4)]*+{\textrm{A}})}
\newgraphescape{P}[1]{[]*++{#1}*\cir{}="P"([d(0.12)r(0.29)]*+{A})}

\newgraphescape{s}[1]{[]*+{#1}*\cir{}="P"([d(0.2)r(0.4)]*+{\textrm{E}})}
\newgraphescape{o}[1]{[]*+{#1}*\cir{}}
\newgraphescape{O}[1]{[]*++{#1}*\cir{}}

\newgraphescape{g}[1]{%
  []!s{v_{#1,1}}(:[l]!o{x_{#1,1}})
  :[d]!p{v_{#1,2}}(:[l]!o{x_{#1,2}})
  :[d]!s{v_{#1,3}}(:[l]!o{x_{#1,3}})
  :[d]!p{v_{#1,4}}(:[l]!o{x_{#1,4}})
  :[d]!s{v_{#1,5}}(:[l]!o{x_{#1,5}})
  "P"="E_{#1}"
}

\newgraphescape{a}[1]{[]*+{#1}*\frm{o}="P"([d(0.2)r(0.4)]*+{\pi})}
\newgraphescape{b}[1]{[]*+{#1}*\frm{o}="P"([d(0.2)r(0.4)]*+{\sigma})}
\newgraphescape{c}[1]{[]*+{#1}*\frm{o}}

\newgraphescape{j}[1]{%
  []!b{a_{#1,1}}(:[l]!c{x_{#1,1}})"P"="B_{#1}"
  :[d]!a{a_{#1,2}}(:[l]!c{x_{#1,2}})
  :[d]!b{a_{#1,3}}(:[l]!c{x_{#1,3}})
  :[d]!a{a_{#1,4}}(:[l]!c{x_{#1,4}})
  :[d]!b{a_{#1,5}}(:[l]!c{x_{#1,5}})
  "P"="E_{#1}"
}

\newcommand{\firstvariable}[1]{\ifthenelse{\equal{#1}{a}}{a_{0}}{b_{0}}}
\newcommand{\secondvariable}[1]{\ifthenelse{\equal{#1}{a}}{a_{1}}{b_{1}}}
\newcommand{\thirdvariable}[1]{\ifthenelse{\equal{#1}{a}}{a_{2}}{b_{2}}}
\newcommand{\firstplayer}[1]{#1}
\newcommand{\secondplayer}[1]{
  \ifthenelse{\equal{#1}{\pi}}{\sigma}{
    \ifthenelse{\equal{#1}{\sigma}}{\pi}{
      \ifthenelse{\equal{#1}{M}}{O}{M}
      }
    }
}

\ifx\TAC\undefined
\RequirePackage{ifthen}
\newcommand{\Whitstar}{}
\newcommand{\Whitvariablea}{a_{0}}
\newcommand{\Whitvariableb}{a_{1}}
\newcommand{\Whitvariablec}{a_{2}}
\newcommand{\Whitlabel}[3]{
  \ifthenelse{\equal{#2}{\pi}}{\renewcommand{\Whitstar}{}}
  {\renewcommand{\Whitstar}{*}}
  \ifthenelse{\equal{#3}{a}}{
   \renewcommand{\Whitvariablea}{a_{0}}
   \renewcommand{\Whitvariableb}{a_{1}}
   \renewcommand{\Whitvariablec}{a_{2}}
   }
  {\renewcommand{\Whitvariablea}{b_{0}}
   \renewcommand{\Whitvariableb}{b_{1}}
   \renewcommand{\Whitvariablec}{b_{2}}}
 W_{#1}^{\Whitstar}(\Whitvariablea,\Whitvariableb,\Whitvariablec)
}
\newgraphescape{A}#1#2#3#4#5{
  []!p{\firstplayer{#2}}="S"
  ( -[l(#4)]-[d(#5)]-[d(#5)]-[r(#4)]!p{\firstplayer{#2}}="B",
  -[r(#4)]-[d(#5)]-[d(#5)]-"B")
  [d(#5)]*{\Whitlabel{#1}{#2}{#3}}
  "B"
}
\newgraphescape{T}#1#2#3{[]!A{#1}{#2}{#3}{0.5}{0.8} }

\fi

\newgraphescape{G}#1{
  []="L"[l(0.2)]
  [r(0.6)u(0.8)]
  ( -[d(0.6)d(0.6)l(0.6)]-[r(0.6)r(0.6)]-[r(0.6)l(0.6)],
   -[d(0.6)d(0.6)r(0.6)])
  [d(0.8)]*+{#1}="C"
 [r(0.4)]="R"
}
\newgraphescape{c}{
  []="A0"
  [l(0.5)d(0.15)]*=<1em>[o][F]{}="A"
  [r]*=<1em>[o][F]{}="B"
  !{"A"!R}[d(0.0)]="AD"
  !{"B"!L}[d(0.0)]="BD"
  "AD"-@`{"A0"+(0,0.15)}"BD"
  "A0"
}
\newgraphescape{d}{
  []="A0"
  [u(0.5)r(0.15)]*=<1em>[o][F]{}="A"
  [d]*=<1em>[o][F]{}="B"
  !{"A"!D}[d(0.0)]="AD"
  !{"B"!U}[d(0.0)]="BD"
  "AD"-@`{"A0"+(-0.15,0)}"BD"
  "A0"
}
\newgraphescape{e}{
  []="A0"
  [u(0.5)l(0.15)]*=<1em>[o][F]{}="A"
  [d]*=<1em>[o][F]{}="B"
  !{"A"!D}[d(0.0)]="AD"
  !{"B"!U}[d(0.0)]="BD"
  "AD"-@`{"A0"+(0.15,0)}"BD"
  "A0"
}

\newgraphescape{E}[1]{
  []*{#1}*+++[o]{}*+\frm{o}="#1"
}

\makeatletter
\renewcommand{\thefigure}{\ifnum \c@section>\z@ \thesection.\fi
 \@arabic\c@figure}
\@addtoreset{figure}{section}
\makeatother

\begin{document}

\title{From Causality Semantics to Duration Timed Models}


\numberofauthors{1}
\author{
%
Walid Belkhir\\
       \affaddr{Laboratoire d'Informatique Fondamentale de Marseille }\\
       \affaddr{}\\
       \affaddr{Universit\'e de Provence, Marseille, France.}\\
       \email{belkhir@cmi.univ-mrs.fr}
}

\date{30 July 1999}

\maketitle
\begin{abstract}
 The interleaving semantics is not compatible with both action refinement and durational  actions. Since many true concurrency semantics are congruent w.r.t. action refinement, notably the causality and the maximality ones \cite{mahboulthese,Glabthese}, this has challenged us to study the dense time behavior  - where the actions are of \emph{arbitrary fixed duration} -  within the   causality semantics of Da Costa \cite{mahboulthese}.     
 
We extend the causal transition systems with the  clocks and the timed constraints, and thus we obtain an over class of timed automata  where the actions  need not to be atomic.   We define a real time extension  of  the formal description technique CSP,  called \emph{duration-CSP}, by attributing  the duration to actions.  We give  the operational timed  causal  semantics of  duration-CSP  as well as its denotational semantics over the class of timed causal transition systems.   Afterwards, we prove that the two semantics are equivalent.  Finally we extend the duration-CSP language with a refinement operator $\rho$ - that allows to replace an action with a process - and  prove that it preserves the timed causal bisimulation.  
\end{abstract}



\keywords{Causality semantics, true concurrency, durational process algebra, timed automata, action refinement.   } 
\section{Introduction}

Many complex   systems such as communication protocols, networks and embedded systems require a top down design where processes are modeled at different levels of abstraction. To carry on, at every level of abstraction, each action might be replaced by a more complicated process. This is known as the concept of \emph{action refinement} \cite{CourtiatS93,Saidouni94syntacticaction,FecherMW02,e-lotos:ref}. It turns out that the actions are no longer atomic:  they are divisible into small parts. On the other hand, many industrial systems exhibit \emph{quantitative} behaviour, including timing and minimal performance. As a consequence, many real time extensions have been suggested for   process algebra \cite{Moller90,ET-LOTOS,Baeten00,Yi91}. However,  the common point of all these extensions is that they are  based on the action atomicity hypothesis.  It was pointed out \cite{Glabthese,mahboulthese,saidounithese} that the non atomicity of actions as well as the action refinement  require a truth concurrency semantics instead of the interleaving semantics. 

In this paper we suggest an approach that integrates both the timed constraints and  durational actions  without replacing    the action with the   two atomics events:   its starting and finishing ones, which leads to a huge combinatorial explosion.   Our approach consists in  using a truth concurrency semantics called the \emph{timed causal  semantics} which extends the causality semantics of \cite{mahboulthese}.  We extend the formal description technique CSP with both  durational actions and timed constraints. Afterwards we describe its semantics   by means of the timed causal semantics. To convince the reader that the interleaving semantics can not be used to deal with the durational actions,  let us consider the two processes $P=a;b;stop \;+\; b;a;stop$ and $Q=a; stop \;||| \;  b;stop$.   The process $P$ expresses a choice between $a$ followed by $b$ and $b$ followed by $a$. The process $Q$ expresses a parallel execution of $a$ and $b$. Note that, if we consider that  $duration(a)=0$ and $duration(b)=0$, then the two processes describe, in some sense, the same behavior. However, if we consider that  $duration(a) > 0$ and $duration(b)>0$  then the execution of $P$ requires at least an amount of times equals to $duration(a)+duration(b)$,  and the execution of $Q$ may be done in $max\set{duration(a),duration(b)}$.  
  
In a next step we extend the causal transition systems of \cite{mahboulthese} with clocks and timed constraints in the same spirit of the timed automata \cite{Alur94}. We shall call this model \emph{the timed causal transition system}.  We recall that  the causal transition system formalism  enriches the usual transition system one  with the notion of causality. As a consequence the timed causal transition system formalism allows to express the timed constraints over the actions of arbitrary duration without the need of replacing  each action by its starting and finishing event. As an application we show how to generate a timed causal transition system out of a duration-CSP process, and prove the correctness of this generation.

The paper is organized as follows. Section \ref{causality:sec} recalls the rudiments of the causality semantics as given in \cite{mahboulthese}. In section \ref{T-CTS:sec} the definition of the causal transition system formalism and its timed extension are given.  In section \ref{D-CSP:sec} we extend the kernel of CSP with action duration and timed constraints and we give its timed causal operational semantics.  In section \ref{Den:sec} we give the denotational semantics of duration-CSP in terms of the timed causal transition system model. This section is concluded by  a proof that the two semantics are equivalent, Theorem  \ref{main:theorem}.  In section \ref{ref:sec} we enrich the language duration-CSP with the  refinement operator $\rho$ that allows to replace an action with a more complicated process. The new language is called $duration-CSP_{\rho}$, afterwards, we give the timed causal semantics of this language, notably, the semantics of the refinement operator. Finally we prove that the refinement operator preserves the timed causal bisimulation, Theorem \ref{ref:th}.  \\

In section \ref{conc:sec} some current and future  works are given.  \longversion{Due the lack of space,  the proofs are given in the Appendix.}   The proofs are given in the Appendix.

\section{Causality semantics}\label{causality:sec}
In this section we recall, through simple examples,  the principles of the causality semantics as defined in \cite{mahboulthese}.  The aim of the causality semantics is to distinguish between the sequential and the parallel execution.  To be more precise, a parallel execution of two actions can not be substituted by their  interleaved execution. To this goal, a transition from state $s_1$ to $s_2$  has the  form $s_1 \stackrel{_{E}a_x}{\fleche} s_2$; it is  equipped with an extra data: \emph{(i)} the event $x$ which identifies the beginning of the execution of the action $a$, and \emph{(ii)} the (finite) set $E$ of events which corresponds to the set of causes of the action $a$, i.e. the action $a$ is possible if all the causes belonging to $E$ terminate. For example let us consider the two  processes $P$ and $Q$ defined by: $P= a;b;stop  \;\ou\;  b;a;stop$ and  $Q= a;stop \; |||\; b;stop$. We recall that  $";"$ is the prefixing operator,  $"||| "$ is the parallel composition , and $"+"$ is the choice operator.  At the  beginning, the execution  of both $P$ and $Q$ does not depend on any event, therefore the initial configuration associated to $P$ (resp. $Q$)  is  of the form ${}_{\emptyset}[P]$ (resp.  ${}_{\emptyset}[Q]$).  By applying the causality semantics to the configuration $_{\emptyset}[P]$  the  following derivations are possible:  
$$
_{\emptyset}[P] \stackrel{{}_{\emptyset}a_x}{\fleche}\;\;  {}_{\set{x}}[b;stop] \stackrel{_{\set{x}}b_y}{\fleche}\;\; _{\set{y}}[stop]
$$ 
The event $x$ (resp. $y$)  corresponds to the beginning of the execution of the action $a$ (resp. $b$). According to the semantics of the prefix operator $";"$, the execution of the action $b$ depends on the termination of the action $a$. Again, by applying the causality semantics to the configuration $_{\emptyset}[Q]$,  the following derivations are possible:   
  $$
_{\emptyset}[Q] \stackrel{{}_{\emptyset}a_x}{\fleche}\;\;  {}_{\set{x}}[stop] \; ||| \; _{\emptyset}[b;stop]  \stackrel{_{\emptyset}b_y}{\fleche}\;\;  _{\set{x}}[stop]\; |||\;  _{\set{y}}[stop]
$$           
As before,  the event $x$ (resp. $y$)  corresponds to the beginning of the execution of the action $a$ (resp. $b$). The main difference is that both the actions $a$ and $b$ does not depend on each other. 

The Figure \ref{CTS:fig} shows all the possible derivations which can be obtained by applying the causality semantics to $P$ and $Q$. This gives rise to the notion of \emph{causal transition systems} which will be formalized in the next section.    
\begin{figure}
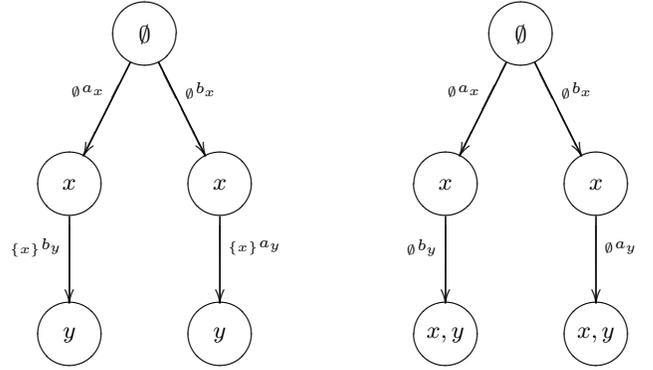

 $$ 
  \xygraph{
   !{<0cm,0cm>;<1cm,0cm>:<0cm,1cm>::}
    []!E{\emptyset}="A"
    (:[ddl]!E{x}="Al"
    :[dd]!E{y}="Ald")
    :[ddr]!E{x}="Ar"
    :[dd]!E{y}="Ard"
      "A"-@{}@/-0.5em/"Al"_{{}_{\emptyset}a_{x}}
       "Al"-@{}@/-0.5em/"Ald"_{{ }_{\{x\}}b_{y}}
      "A"-@{}@/-0.5em/"Ar"^{{}_{\emptyset}b_{x}}
       "Ar"-@{}@/-0.5em/"Ard"^{{ }_{\{x\}}a_{y}}
    "A"[u(1)]*+{P= a;b;stop \; + \; b;a;stop}
    "A"[r(3)]([d(1.5)]*+{}) [r(2)]
        []!E{\emptyset}="B"
    (:[ddl]!E{x}="Bl"
    :[dd]!E{{x, y}}="Bld")
    :[ddr]!E{{x}}="Br"
    :[dd]!E{{x,y}}="Brd"
      "B"-@{}@/-0.5em/"Bl"_{{}_{\emptyset}a_{x}}
       "Bl"-@{}@/-0.5em/"Bld"_{{ }_{\emptyset}b_{y}}
      "B"-@{}@/-0.5em/"Br"^{{}_{\emptyset}b_{x}}
       "Br"-@{}@/-0.5em/"Brd"^{{ }_{\emptyset}a_{y}}
    "B"[u(1)]*+{Q= a;stop \;||| \; b; stop}
 }
  $$
\caption{Causal transistion systems of the processes $P$ and $Q$.}
\label{CTS:fig}
\end{figure}

\section{Timed causal transition systems}\label{T-CTS:sec}
In this section we formalize the notion of causal transition systems. Afterwards,   we enrich  them with clocks  and timed constraints  in order to specify  the timed behaviour. Throughout this paper we let $\evns$ be a countable set of \emph{events},  ranged by $x,y,z \dots$. Let $\actset$ be a countable set of \emph{actions},  ranged by $a,b,c, \dots$. If $a \in \actset$ then we denote by $\dr{a}$ the \emph{duration} of the action $a$,   where $\dr{a}\in \dom$. 

\begin{definition}\label{causal:trans:sys:def}
 A \myemph{causal transition system}, or a \myemph{CTS} for short,  over $\evns$ is a tuple $(S,s_0,T,l,\psi,\zeta,\eta)$ where: 
\begin{itemize}
\item $(S,s_0,T,l)$ is a labeled transition system over $\actset$, that is, $S$ is a finite set of states,  $s_0\in S$ is the initial state, $T \subseteq S\times S$ is the set of transitions, and $l: T \rTo \actset$ is the  labeling function of transitions,    
\item $\psi: S \rTo 2^{\evns}$ is the function that associates to each state a finite set of events,  the latter being potentially in progress at this state,  
\item $\zeta: T \rTo 2^{\evns}$ is the function that associates to each transition $t \in T$ a finite set of events, these events denote the \myemph{direct causes} of $t$,
\item    $\eta: T \rTo  \evns$ is the function that associates to each transition $t\in T$ the event attached to the occurrence of the action $l(t)$,    
\end{itemize}  
such that the following conditions hold: for each transition $(s,s')\in T$  we have that
\begin{enumerate}[{i.}]
\item $\eta(s,s')\in \psi(s')$,
\item $\zeta(s,s') \cap ( \psi(s')- \eta(s,s')) = \emptyset$\longversion{, i.e. the terminated events - apart $\eta(s,s')$- should not appear in $s'$},
\item $\zeta(s,s') \subseteq \psi(s)$ and $\psi(s')-\zeta(s,s') \subseteq \psi(s)$\longversion{, i.e. both the terminated and the non-terminated events should be included in $s$}. 
\end{enumerate}
\end{definition}

In the next a transition $t$ will be denoted by $s_1\stackrel{_{E}a_x}{\fleche} s_2$, i.e. $l(t)=a$, $\zeta(t)=E$, and $\eta(t)=x$. 

\parpic[r]{
 $$ 
  \xygraph{
   !{<0cm,0cm>;<0.4cm,0cm>:<0cm,1cm>::}
    []!E{\emptyset}="A"
    :[dd]!E{x}="A1"
    :[dd]!E{y}="A2"
      "A"-@{}@/-0.5em/"A1"_{{}_{\emptyset}a_{x}}
       "A"-@{}@/-0.5em/"A1"^{c_x:=0; \;\; 0 \le c_x \le 4}
       "A1"-@{}@/-0.5em/"A2"_{{ }_{\{x\}}b_{y}}
        "A1"-@{}@/-0.5em/"A2"^{c_y:=0; \;\;  c_x = \dr{a}+100?}
        "A"[u(0.75)r(1.7)]*+{R=a\set{4};  \delay{100}b;stop}
        "A"[d(04.75)r(1.7)]*+{\textrm{Fig 3.1: The timed-CTS of } R}
}
$$
}
 \textbf{The key idea.} Now, we add to the CTS the notions of clocks and timed  constraints in order to be able to specify the quantitative behaviour over durational actions. The key idea consists in considering the events themselves as a sort of local clocks. As a consequence, the values of the  clocks give sufficient information about the progress of the  actions,  notably about their termination. For instance consider the timed  process $R$ defined by $R=a\set{4}; \delay{100}b$ which specifies  that the action $a$ can occur in the  interval $[0,4]$, and the action $b$ can occur after $100$ units of time counting from the termination of $a$.   The timed CTS corresponding to the process $R$ is depicted in Figure 3.1. In order to avoid any confusion, we denote the clock associated to the event $x$ by $c_x$ and not $x$. The semantics of the timed-CTS is close to that of  the timed automata. The construction of the timed-CTS   out of a of  duration-CSP process is  given  in Section \ref{Den:sec}.  \longversion{Note that we can omit  the dependency   relation and substitute it with clock constraints, for instance we can substitute $_{\set{x}}b_y$ by $b_y$ and add the constraint $c_x\ge \dr{a}$ to the $b$-transition, this constraint explicits the termination of action $a$.  This is possible because each action has a \emph{fixed} duration, otherwise, if the duration of some actions belongs to a fixed interval then the timed constrains can not replace the dependency relation, at least in a straightforward way. However,  in the next we keep the dependency relation in order to keep the model as generic as possible.} \\ The  definition of the timed-CTS follows.


\begin{definition}\label{T-CTS:def}
A \myemph{timed causal transition system}, or a \myemph{timed-CTS} for short, is a tuple $(S,T,s_0,l,\psi,\zeta,\eta,\clk,\Phi,\Lambda)$ where $(S,s_0,T,l,\psi,\zeta,\eta)$ is a causal transition system (see Def. \ref{causal:trans:sys:def}) and
\begin{itemize}
\item $\clk = \set{c}\times \evns$ is the set of clocks, that is, to each event $x\in \evns$ we associate a clock $c_x$,
\item $\Phi$ is a function that associates to each transition $t\in T$ a \myemph{timed constraint}, and  
\item $\Lambda: T \rTo 2^{\clk}$ is a function that associates to each transition the set of clocks which have to be reset to zero once this transition is executed.  
\end{itemize} 
\end{definition}
In the next, a transition $t$ of a timed-CTS will be denoted simply by $s_1 \stackrel{\langle _E a_x, \varphi, \lambda \rangle}{\fleche} s_2$, that is, $\Phi(t)=\varphi$ and $\Lambda(t)=\lambda$.
The set of timed constraints will be denoted by $2^{\varphi}$. The syntax of  the timed constraints  is given by the following  grammar: 
\begin{align*}
\varphi  &::= \varphi \myet \varphi \;|\;  \varphi \myou \varphi \;|\; c_x \prec c \;|\; c \prec  c_x \hspace{18mm} \prec \in \set{<, \le} 
\end{align*}
where $c_x $ is a clock,  and $c\in \dom$ is positive real  constant.\\   
  The  timed-CTS inherits the semantics of both  timed automata \cite{Alur94}, and causal transition systems \cite{mahboulthese}. The semantics of a timed-CTS   is defined by means of a transition system over a set of configurations, each  configuration consists of \emph{(i)} the current state, \emph{(ii)} the current values of clocks, and \emph{(iii)} the  actions which are (potentially) in progress.  There are two kinds of transitions between configurations. The timed-CTS may either delay for an amount of time in the same configuration (delay transition), or follow an edge (action transition). \\
We use functions called \emph{clock assignments}, a mapping from $\clk$ to $\dom$. Let $\as$ denote such function, and  $\zero$ denote the clock assignment that maps all $c_x \in \clk$ to $0$. For  $d \in \dom$, let $\as +d$  denote the clock assignment that maps all $c_x \in \clk$ to $\as(c_x)+d$. For $\lambda \subseteq  \clk$, let $[\lambda \mapsto 0]\as$ denote the clock assignment that maps all clocks in $\lambda$ to $0$  and coincide with $\as$ for the clocks in $\clk \setminus \lambda$.    \\

\textbf{The semantics} of a timed-CTS is a transition system whose configurations  are pairs  $\aconf{s}{\as}$, the starting configuration  is $\aconf{s_0}{\zero}$,  and the transitions are given by the rules:
\begin{itemize}
\item $\aconf{s}{\as} \stackrel{d}{\fleche} \aconf{s}{\as + d}$, for $d\in \dom$, 
\item $\aconf{s}{\as} \stackrel{_{E}a_x} {\fleche} \aconf{s'}{\as'}$ if $s \stackrel{\langle _Ea_x, \varphi, \lambda \rangle}{\fleche} s'$ and moreover:  (i) $\as$ satisfies the constraint $\varphi$, (ii)  $\as'=[\lambda \mapsto 0]\as$, and  (iii) all the actions related  to the events $E$ have terminated. 
\end{itemize}

\longversion{
\subsection{Timed bisimulation of timed-CTS}
We define the notion of timed bisimulation between timed-CTS.
 Before this, let us fix some notations. \\
Let $f: A \rTo B$ and let $A' \subseteq A$ and $B' \subseteq B$. The parametrized restrictions of $f$ w.r.t.  its domain and co-domain are defined respectively as follows:
\begin{align*}  
&  f_{\pi_1(A')} := \set{(a,b)  \;| \;a \in A'} \hspace{2mm} \textrm{ and }
&   f_{\pi_2(B')} := \set{(a,b) \;|\; b \in B'}
\end{align*}

\begin{definition}\label{tcsbis:def}
  \myemph{A timed causal bisimulation} over the "states" of a timed-CTS  $\tcsbis$  is a  binary relation that comes with an events' bijection  $f: E \flechep E$, and satisfies the following conditions:
\begin{enumerate}[{1}.1.] 
\item if $\aconf{s_1}{\nu_1} \stackrel{_Ea_x} {\fleche} \aconf{s_1'}{\nu'_1}$ then there exists \\ $\aconf{s_2}{\nu_2} \stackrel{_Fa_y} {\fleche} \aconf{s_2'}{\nu'_2}$ such that 
\begin{enumerate}[{i.}]
\item  $u\in E$ if and only if  $f(u)\in F$, and
\item $(\aconf{s_1'}{\nu'_1},\aconf{s_2'}{\nu'_2})_{f'} \in  \tcsbis $  \\ where $f':= (f_{\pi_1(\psi(s_1')-x )})_{\pi_2(\psi(s_2')-y )} \; \cup \set{(x,y)}$.
\end{enumerate}
\item if $\aconf{s_1}{\nu_1} \stackrel{d} {\fleche} \aconf{s_1}{\nu'_1}$ then $\aconf{s_2}{\nu_2} \stackrel{d} {\fleche} \aconf{s_2}{\nu'_2}$ and $(\aconf{s_1}{\nu'_1},\aconf{s_2}{\nu'_2})_{f} \in  \tcsbis $.
 \end{enumerate}
\begin{enumerate}[{2}.1.] 
\item if $\aconf{s_2}{\nu_2} \stackrel{_Fa_y} {\fleche} \aconf{s_2'}{\nu'_2}$ then there exists \\ $\aconf{s_1}{\nu_1} \stackrel{_Ea_x} {\fleche} \aconf{s_1'}{\nu'_1}$ such that 
\begin{enumerate}[{i.}]
\item  $u\in E $ if and only if  $f(u)\in F$, and
\item $(\aconf{s_1'}{\nu'_1},\aconf{s_2'}{\nu'_2})_{f'} \in  \tcsbis $  where \\ $f':= (f_{\pi_1(\psi(s_1')-x )})_{\pi_2(\psi(s_2')-y )} \;\cup \set{(x,y)}$.
\end{enumerate}
\item if $\aconf{s_2}{\nu_2} \stackrel{d} {\fleche} \aconf{s_2}{\nu'_2}$ then $\aconf{s_1}{\nu_1} \stackrel{d} {\fleche} \aconf{s_1}{\nu'_1}$ and  $(\aconf{s_1}{\nu'_1},\aconf{s_2}{\nu'_2})_{f} \in  \tcsbis $.
 \end{enumerate}
Two timed-CTS are  bisimilar iff there exists a timed causal bisimulation containing their initial configurations.
\end{definition}
}

\section{Duration-CSP and its operational timed causal semantics}\label{D-CSP:sec}

Now we introduce the action duration to  the formal description technique CSP \cite{Hoar85CSP}. Due to the lack of space the prefixing operator $"\flechep"$ is denoted by $";"$ . Moreover, we do not distinguish between the internal and the external choice.  The syntax of duration-CSP is given by the following grammar:
\begin{align*}
P ::= & \; stop \;|\; skip\set{d} \;|\;   \delay{d}P \;|\; a\set{d};P \;|\; P \ou Q \;|\; P|[L]|Q \;|\;  \\ & P \setminus L  \;|\; P \inter Q  
\end{align*}
where $d \in \dom$ and  $L \subseteq \actset$.\\ 
 The primitive process  $stop$  represents the  process that communicates nothing,  and $skip$  represents successful termination i.e. the process $skip\set{d}$ performs the successful termination action $\delta$ in the time interval $[0,d]$ and transforms into $stop$. 
Let $a\in \actset$ be an action and $d\in \dom$. The process $a\set{d}; P $ expresses  that the execution of $a$ must  be in the time interval $[0,d]$, and after the termination of $a$ this process behaves like $P$. The process $\delay{d}P$ means that the starting of  $P$ is possible only  after a passage of  $d$ units of time.  "+"    is the  choice operator.  The parallel composition $P |[L]|  Q$ allows computation in $P$ and $Q$ to proceed simultaneously and independently apart on the actions in $L$ on which  both processes must be synchronized. We shall write $|||$ for $|[\emptyset]|$.    The hiding operator  $P \setminus L$  makes the  actions in $L$ unobservable. The interruption operator $P \inter Q$ allows the computation to begin in $P$ and to be interrupted by $Q$.
    
\paragraph{Operational semantics of Duration-CSP} ~\\
Now we describe the behaviour of duration-CSP processes  step by step by means of the operational semantics over the \emph{timed causal configurations}. Before this, we first  define the  timed causal configurations and introduce some standard operations on them.  The untimed configurations and the related  operations have been defined in \cite{mahboulthese}.   
\begin{definition}\label{timed-causal-config:def}
The set $\tconf$  of \myemph{timed causal configurations} is defined as follows:
\begin{itemize}
\item for each duration-CSP process $P$ and for each \\ $E_{\tau} \in 2^{\evns \times \actset \times \dom}$, we have that $_{E_{\tau}}[P] \in \tconf$,
\item if $\mat{P}_{\tau} \in \tconf$ then $\delay{d} \mat{P}_{\tau} \in \tconf$, for every $d\in \dom$,
\item if $\mat{P}_{\tau} \in \tconf$ then $\mat{P}_{\tau} \setminus L \in \tconf$, and
\item if $\mat{P}_{\tau}, \mat{Q}_{\tau} \in \tconf$ then $\mat{P}_{\tau} \otimes \mat{Q}_{\tau} \in \tconf$, where \\ \hspace{13mm} $\otimes \in \set{\ou\; ,\; |[L]| \;, \; \inter}$.
\end{itemize} 
\end{definition}
For instance, the configuration $_{\set{x:a:t_x}}[P]$  means that the execution of the process $P$ depends on the termination of the action $a$ which is identified by the event $x$, moreover,  $t_x$ counts the  time elapsed   from the beginning of $a$.  We say that a timed causal configuration is in the \emph{canonical} form if it can not be simplified  by distributing  the set of events over the algebraic operators. For instance, the configuration $_{E_{{\tau}}}[a;stop \ou b;stop]$ is not in the canonical form because it can be reduced to the configuration $_{E_{\tau}}[a;stop] \; \ou _{E_{\tau}}[b;stop]$, the latter being in the  canonical form.
\begin{lemma}\label{canonical:lemma}
Every canonical  timed causal configuration in  $\tconf$  has one of the following forms:
\begin{align*}
&  _{E_{\tau}}[stop]   \hspace{5mm} _{E_{\tau}}[skip\set{d}]  \hspace{5mm} \delay{d} \mat{P}_{\tau}  \hspace{5mm} _{E_{\tau}}[a\set{d};P] 
 \hspace{5mm} \\ 
&  \mat{P}_{\tau} \ou \mat{Q}_{\tau}  \hspace{5mm} \mat{P}_{\tau}|[L]|\mat{Q}_{\tau}  \hspace{5mm} \mat{P}_{\tau} \setminus L  \hspace{5mm} \mat{P}_{\tau} \inter \mat{Q}_{\tau} 
\end{align*}
where $\mat{P}_{\tau}$ and $\mat{Q}_{\tau}$ are in the canonical  form. 
\end{lemma}

Next  we assume that all the configurations are in the canonical form.
\begin{definition}\label{encap:def}
The function $\psi : \tconf \flechep 2^{\evns \times \actset \times \dom }$, that determines  the events of a given configuration is defined by:
\begin{align*}
&  \psi(_{E_{\tau}}[stop])  =    \psi( _{E_{\tau}}[skip\set{d}])  =\psi(_{E_{\tau}}[a\set{d};P]) = E_{\tau} \;\;\;\\  
&  \psi(\delay{d} \mat{P}_{\tau})  =  \psi(\mat{P}_{\tau} \setminus L)=\psi(\mat{P}_{\tau})\\ 
 &   \psi(\mat{P}_{\tau} \ou \mat{Q}_{\tau})= \psi(\mat{P}_{\tau}|[L]|\mat{Q}_{\tau}) = \psi(\mat{P}_{\tau} \inter \mat{Q}_{\tau})= \psi(\mat{P}_{\tau}) \cup \psi(\mat{Q}_{\tau}) 
\end{align*} 
\end{definition}


\begin{definition}
Let $\mat{R}_{\tau} \in \tconf$ and $x,y \in \evns$, the substitution of $x$ by $y$ in $\mat{R}_{\tau}$, denoted by $\mat{R}_{\tau}[y/x]$, is defined by induction on $\mat{R}_{\tau}$ as follows: 
\begin{align*}
    (_{E_{\tau}}[stop])[y/x]& =_{E_{\tau}[y/x]}[stop]   \\  (_{E_{\tau}}[skip\set{d}])[y/x]&=_{{E_{\tau}}[y/x]}[skip\set{d}]  \\ 
 (\delay{d} \mat{P}_{\tau})[y/x]& = \delay{d} (\mat{P}_{\tau}[y/x])  \\   (_{E_{\tau}}[a\set{d};P])[y/x]&=_{{E_{\tau}}[y/x]}[a\set{d};P] \\   
  (\mat{P}_{\tau} \ou \mat{Q}_{\tau})[y/x] &= \mat{P}_{\tau}[y/x] \ou\mat{Q}_{\tau}[y/x] 
     \\ (\mat{P}_{\tau} \setminus L) &=\mat{P}_{\tau}[y/x] \setminus L  \\  
 (\mat{P}_{\tau}|[L]|\mat{Q}_{\tau})[y/x] & = \mat{P}_{\tau}[y/x]|[L]|\mat{Q}_{\tau}[y/x] \\  (\mat{P}_{\tau} \inter \mat{Q}_{\tau})[y/x] &=\mat{P}_{\tau}[y/x] \inter \mat{Q}_{\tau}[y/x]  
\end{align*}
 where $E_{\tau}[y/x]$ is again  the obvious substitution over the  set of events. 
\end{definition}

Let $E_{\tau} \in 2^{\evns \times \actset \times \dom}$, we say that all the actions in  $E_{\tau}$ have finished and write $Finish(E_{\tau})$, if for all $x:a:t_x \in E_{\tau}$ we have that $t_x > \dr{a}$. Let $get : 2^{\evns} \flechep \evns$ be a function  satisfying $get(E)\in E$, $\forall E \in 2^{\evns}-\set{\emptyset}$.  \\

The timed transition over the timed causal configurations, denoted by $\flecheop\;  \subseteq \tconf \times Act_{\tau} \times \tconf$ where $Act_{\tau} = (2^{\evns \times \actset \times \dom} \times \actset \times  \evns ) \; \cup \; \dom$,  is defined as follows:

\noindent 0. \textbf{$Stop$ process}: 
$$
\frac{\neg Finish(E_{\tau})}
{ {}_{E_{\tau}}[stop]  \stackrel {d}{\flecheop} {{}_{E_{\tau}+d}[stop]}}
$$
I. \textbf{$Skip$ process}:

$$\textrm{(I.a)}\; \frac{Finish(E_{\tau})}
{  {_{E_{\tau}}[skip\set{u}] }  \stackrel{ _{E_{\tau}}\delta_x}{\flecheop} { } _{\set{x:\delta:0}}[stop]} 
$$
$$\textrm{(I.$\tau$)} \; \frac{Finish(E_{\tau})}
{  {_{E_{\tau}}[skip\set{d+d'}] }  \stackrel{d'}{\flecheop} { } _{E_{\tau}}[skip\set{d}]} 
$$

\vspace{4mm}
\noindent II. \textbf{Prefix operator}:

$$\rm{(II.a)}\;
\frac{ Finish(E_{\tau})}
{ _{E_{\tau}}[a\set{u}; P]  \stackrel{_{E_{\tau}}a_x}{\flecheop } \; {}_{\set{x:a:0}}[P] }
\quad x=get(\mat{\evns})
$$
$$\textrm{(II.$\tau$)} \;  \frac{Finish(E_{\tau})}
{  {_{E_{\tau}}[a\set{d+d'}; P] }  \stackrel{d'}{\flecheop}  {}_{E_{\tau}}[a\set{d}; P]} 
$$

\vspace{2mm}

\noindent III. \textbf{Choice operator}:
$$
\textrm{(III.a)}\;
\frac
{\mat{P}_{\tau} \stackrel{_{E_{\tau}}a_x} {\flecheop}\mat{P}_{\tau}'}
{\mat{P}_{\tau} \ou \mat{Q}_{\tau} \stackrel{_{E_{\tau}}a_x} \flecheop \mat{P}_{\tau}' \hspace{15mm} 
  \mat{Q}_{\tau} \ou \mat{P}_{\tau} \stackrel{_{E_{\tau}}a_x} \flecheop \mat{P}_{\tau}' }
$$
$$
\textrm{(III.$\tau$)}\;
\frac{
\mat{P}_{\tau} \stackrel{d}{\flecheop} \mat{P}_{\tau}' \hspace{10mm} \mat{Q}_{\tau} \stackrel{d}{\flecheop} \mat{Q}_{\tau}'
}
{ \mat{P}_{\tau} \ou \mat{Q}_{\tau} \stackrel{d} {\flecheop} \mat{P}_{\tau}'  \ou \mat{Q}_{\tau}' }
$$
\vspace{3mm}
\noindent IV. \textbf{Parallel composition operator}:
\begin{align*}
  & \textrm{(IV.$\tau$)}
\frac{
\mat{P}_{\tau} \stackrel{d}{\flecheop} \mat{P}_{\tau}' \hspace{8mm} \mat{Q}_{\tau} \stackrel{d}{\flecheop} \mat{Q}_{\tau}'
}
{ \mat{P}_{\tau} |[L]| \mat{Q}_{\tau} \stackrel{d} {\flecheop} \mat{P}_{\tau}'  |[L]| \mat{Q}_{\tau}' }
\\
& \textrm{(IV.a)}
\frac
 {\mat{P}_{\tau} \stackrel{_{E_{\tau}}a_x} {\flecheop} \mat{P}_{\tau}' \quad a \notin L\cup \set{\delta}}
 {\mat{P}_{\tau}|[L]|\mat{Q}_{\tau} \stackrel{_{E_{\tau}}a_y} {\flecheop}  \mat{P}_{\tau}' [y / x] |[L]| \mat{Q}_{\tau} }
\\
& \textrm{(IV.b)}
\frac
 {\mat{P}_{\tau} \stackrel{_{E_{\tau}}a_x} {\flecheop} \mat{P}_{\tau}' \quad a \notin L \cup \set{\delta}}
 {\mat{Q}_{\tau}|[L]|\mat{P}_{\tau} \stackrel{_{E_{\tau}}a_y} {\flecheop}  \mat{Q}_{\tau} |[L]| \mat{P}_{\tau}' [y / x] }
\end{align*}

where in the last two rules we have \\ $ y=get\Big(\mat{\evns} - \big(( \evn{\mat{Q}_{\tau}'} - \set{x} ) \cup \evn{\mat{P}_{\tau} } \big)\Big)$. To avoid any confusion with the definition of $\psi$ given in Definition \ref{encap:def},  here we consider that $\psi : \tconf \flechep 2^{\evns}$ but we still use the same symbol, the type  of $\psi$ is clarified by the context. 
 \begin{align*} 
\textrm{(IV.c)}\;
\frac
 {\mat{P}_{\tau} \stackrel{_{E_{\tau}}a_x} {\flecheop} \mat{P}_{\tau}'  \;\;\; \mat{Q}_{\tau} \stackrel{_{F_{\tau}}a_y} {\flecheop} \mat{Q}_{\tau}' \;\;\;\;\; a \in L\cup \set{\delta}}
 {\mat{P}_{\tau}|[L]|\mat{Q}_{\tau} \stackrel{_{E_{\tau}\cup F_{\tau}}a_z} {\flecheop}  \mat{P}_{\tau}' [z / x] |[L]| \mat{Q}_{\tau}' [z/y] } 
 \end{align*}
 $\;\;\; z= get\Big(\evns - \big[ \big( \evn{\mat{P}'} - \set{x}\big)  \cup \big(\evn{\mat{Q}'} - \set{y} \big)  \big] \Big)$

\noindent V. \textbf{Hide operator}:
$$
\textrm{(V.a)}\;
\frac
{\mat{P}_{\tau} \stackrel{_{E_{\tau}}a_x} {\flecheop} \mat{P}_{\tau}' \;\;\; a \notin L}
{\mat{P}_{\tau} \setminus L \stackrel{ _{E_{\tau}}a_x} {\flecheop} \mat{P}_{\tau}' \setminus L}
\hspace{10mm}
\textrm{(V.b)}\;
\frac
{\mat{P}_{\tau} \stackrel{_{E_{\tau}}a_x} {\flecheop} \mat{P}_{\tau}' \;\;\; a \in L}
{\mat{P}_{\tau} \setminus L \stackrel{_{E_{\tau}}i_x} {\flecheop} \mat{P}_{\tau}' \setminus L}
\hspace{10mm}
$$
$$
\textrm{(V.$\tau$)}\;
\frac
{\mat{P}_{\tau} \stackrel{d} {\flecheop} \mat{P}_{\tau}' }
{\mat{P}_{\tau} \setminus L \stackrel{d} {\flecheop} \mat{P}_{\tau}' \setminus L}
$$

\noindent VI.  \textbf{Interruption operator}:
\begin{align*}
\hspace{4mm}
\textrm{(VI.a)}\;
\frac
 {\mat{P}_{\tau} \stackrel{_{E_{\tau}}a_x} {\flecheop} \mat{P}_{\tau}' \quad \quad a \neq \delta}
 {\mat{P}_{\tau} \inter \mat{Q}_{\tau} \stackrel{_{E_{\tau}}a_y} {\flecheop}  \mat{P}_{\tau}' [y / x] \inter \mat{Q}_{\tau}}
\end{align*}
$\hspace{3mm} y=get \big( \evns - \big[( \evn{\mat{P}_{\tau}' }-\set{x})  \cup \evn{\mat{Q}_{\tau}}\big]\big) 
$

\begin{align*}
\textrm{(VI.b)}\;
\frac
 {\mat{P}_{\tau} \stackrel{_{E_{\tau}}{\delta}_x} {\flecheop} \mat{P}_{\tau}' }
 {\mat{P}_{\tau} \inter \mat{Q}_{\tau} \stackrel{_{E_{\tau}}{\delta}_x} {\flecheop}  \mat{P}_{\tau}' }
\hspace{5mm}
\textrm{(VI.c)}\;
\frac
 {\mat{Q}_{\tau} \stackrel{_{E_{\tau}}a_x} {\flecheop} \mat{Q}_{\tau}' }
 {\mat{P}_{\tau} \inter \mat{Q}_{\tau} \stackrel{_{E_{\tau}}a_x} {\flecheop}  \mat{Q}_{\tau}'  } 
\end{align*}
\begin{align*}
\hspace{5mm}
\textrm{(VI.$\tau$)}\;
\frac{
\mat{P}_{\tau} \stackrel{d}{\flecheop} \mat{P}_{\tau}' \hspace{8mm} \mat{Q}_{\tau} \stackrel{d}{\flecheop} \mat{Q}_{\tau}'
}
{ \mat{P}_{\tau}\inter  \mat{Q}_{\tau} \stackrel{d} {\flecheop} \mat{P}_{\tau}'  \inter \mat{Q}_{\tau}' }
\end{align*}

\vspace{2mm}

\noindent VII. \textbf{Delay operator}:
$$
\textrm{(VII.$\tau$)}\;
\frac{}
{\delay{d+d'}\mat{P}_{\tau} \stackrel{d'}{\flecheop} \delay{d} \mat{P}_{\tau}} 
\hspace{15mm}
\textrm{(VII.$\tau'$)}\;
\frac {\mat{P}_{\tau} \stackrel{d}{\flecheop} \mat{P}_{\tau}'}
{\delay{0}\mat{P}_{\tau} \stackrel{d}{\flecheop}  \mat{P}_{\tau}'} 
$$
$$
\hspace{15mm}
\textrm{(VII.a)}\;
\frac{ \mat{P}_{\tau} \stackrel{{}_{E_{\tau}}a_x}{\flecheop} \mat{P}_{\tau}' }
{\delay{0}\mat{P}_{\tau} \stackrel{{}_{E_{\tau}}a_x}{\flecheop}  \mat{P}_{\tau}'}
$$

\noindent VIII. \textbf{Passage of time}:
$$
\frac{\neg Finish(E_{\tau}) \;\; and \quad \forall \epsilon\;\; 0\le \epsilon \le d \quad \neg Finish(E_{\tau}+ \epsilon)}
 {_{E_{\tau}}[P] \stackrel{d}{\flecheop}\; _{E_{\tau}+d}[P]}
$$

\begin{definition}
Let $\mat{R}_{\tau} \in \tconf$, the passage  of $d$ units of time over $\mat{R}_{\tau}$, denoted by $\mat{R}_{\tau}+d$, is defined by induction on $\mat{R}_{\tau }$ as follows: 
\begin{align*}
 _{E_{\tau}}[P]+d&=_{E_{\tau}+d}[P] \\
 (\mat{P}_{\tau}\ou \mat{Q}_{\tau})+d&=(\mat{P}_{\tau}+d)\ou (\mat{Q}_{\tau}+d) \\  (\mat{P}_{\tau}\setminus L)+d &= (\mat{P}_{\tau}+d)\setminus L  \\
 (\mat{P}_{\tau}|[L]| \mat{Q}_{\tau})+d &=(\mat{P}_{\tau}+d)|[L]| (\mat{Q}_{\tau}+d)
 \\
 (\mat{P}_{\tau}\inter \mat{Q}_{\tau})+d& =(\mat{P}_{\tau}+d) \inter (\mat{Q}_{\tau}+d)  
\end{align*}
where\\  
\begin{math}
\textrm{ } \;\;\;
  \begin{cases} 
\emptyset +d &= \emptyset,  \\
(x:a:t_x)+d &= x:a:t_{x}+d,  \\
(E_{\tau}\cup \set{x:a:t_{x}})+d&=(E_{\tau}+d)\cup \set{(x:a:t_{x})+d}.
  \end{cases}
\end{math}
\end{definition}

\begin{definition}
Given a duration-CSP process $P$, the operational semantics of $P$ over the class of the timed causal configurations $\tconf$,  denoted by $P^{op}$, consists in associating to $P$ the set of timed causal configurations generated by the relation $\flecheop \; \in  \tconf \times Act_{\tau} \times \tconf $, starting from the configuration $_{\emptyset}[P]$.
\end{definition}


\section{A denotational  semantics}\label{Den:sec}
In this section we describe how to generate a timed-CTS (see Definition \ref{T-CTS:def}) from  a duration-CSP  specification.  To this goal, we shall define  the timed causal transition relation $\fleche  \; \subseteq\st  \times trs \times \st$, where $\st$ is defined exactly as the set of the timed configurations $\tconf$ given in  Definition \ref{timed-causal-config:def}, apart that $E_{\tau} \in 2^{\evns \times \actset}$ instead of $E_{\tau} \in 2^{\evns \times \actset  \times \dom}$ and hence $E_{\tau}$  will be denoted by $E$; and the timed  transition $trs  \in (2^{\evns \times \actset }  \times \actset \times \evns ) \times 2^{\varphi} \times 2^{\clk}$. We recall that $2^{\varphi}$ is the set of timed constraints. 
\begin{enumerate}
\item \textbf{Skip  process}:
\begin{align*}
&(1.a) 
 \frac{}{{}_{\emptyset}[skip \set{u}] \stackrel {\langle {}_{\emptyset}\delta_{x}, \; 0\le c_x\le u , \; c_x  \rangle}{\longrightarrow} {}_{\set{x:\delta}}[stop]} \; x=get(\evns)\\
& (1.b)
\frac{E \neq \emptyset}{{}_{E}[skip \set{u}] \stackrel {\langle {}_{E}\delta_{x},\; \fin{\le u}{E},\;c_x \rangle}{\longrightarrow} {}_{\set{x:\delta}}[stop]} \; x=get(\evns)
\end{align*}

\item \textbf{Prefix operator}: \\
$$ 
 (2.a) \frac{}
{{}_{\emptyset}[a \set{u}; P] \stackrel {\langle {}_{\emptyset} a_{x}, \;  0\le c_x\le u, \; c_x  \rangle}{\longrightarrow} {}_{\set{x:a}}[P]} 
\; x=get(\evns)
$$
 $$
(2.b) \frac{E \neq \emptyset}
{{}_{E}[a \set{u}; P] \stackrel {\langle {}\; _{E} a_{x}, \; \fin{\le u}{E}, \; c_x  \rangle}{\longrightarrow} {}_{\set{x:a}}[P]} 
\; x=get(\evns)
$$

\item \textbf{Choice operator}:
 $$ (3.a) \frac
{ \mat{P} \stackrel{\langle trs \rangle}{{\longrightarrow}} \mat{P}' }
{\mat{P} \; \ou \; \mat{Q}  \stackrel{\langle trs \rangle}{{\longrightarrow}} \mat{P}'}
\;\;\;\;\;\;\;\;\;\;\;
  (3.b) \frac
{ \mat{Q} \stackrel{\langle trs \rangle}{{\longrightarrow}} \mat{Q}' }
{\mat{P} \; \ou \; \mat{Q}  \stackrel{\langle trs \rangle}{{\longrightarrow}} \mat{Q}'}$$

\item \textbf{Parallel composition operator}:

 $$ \hspace{-15mm}
(4.a)\;
\frac
  {\mat{P} \stackrel{\langle _{E}a_x,\varphi,\lambda \rangle} {\fleche} \mat{P}' \;\; a \notin L\cup \set{\delta}}
 {\mat{P}\;|[L]|\mat{Q} \stackrel{\langle _{E}a_y,\varphi[c_y/c_x], \lambda[c_y/c_x]\rangle} {\fleche}  \mat{P}' [y / x] \;|[L]|\; \mat{Q} }
$$
$ y=get\Big(\evns - \big(( \evn{\mat{P}'} - \set{x} ) \cup \evn{\mat{Q} } \big)\Big)
$

$$\hspace{-15mm}
(4.b)\;
\frac
 {\mat{Q} \stackrel{\langle _{E}a_x,\varphi,\lambda \rangle} {\fleche} \mat{Q}' \;\; a \notin L \cup \set{\delta}}
 {\mat{P}|[L]|\mat{Q} \stackrel{\langle _{E}a_y,\varphi[c_y/c_x], \lambda[c_y/c_x]\rangle} {\fleche}  \mat{P} \; |[L]|\; \mat{Q}' [y / x] }
$$
$  y=get\Big(\evns - \big(( \evn{\mat{Q}'} - \set{x} ) \cup \evn{\mat{P} } \big)\Big)
$

$$
(4.c)\;
\frac
 {\mat{P} \stackrel{\langle _{E}a_x,\varphi_1,\lambda_1 \rangle} {\fleche} \mat{P}'  \;\;\; \mat{Q} \stackrel{\langle _{F}a_y,\varphi_2,\lambda_2 \rangle} {\fleche} \mat{Q}' \;\;\; a \in L\cup \set{\delta}}
 {\mat{P}\;|[L]|\;\mat{Q} \stackrel{\langle _{E\cup F}a_z ,\; \Omega,\; \Gamma \rangle} {\fleche}  \mat{P}' [z / x] \;|[L]| \;\mat{Q}' [z/y] } 
 $$
$  z= get\Big(\evns - \big[ \big( \evn{\mat{P}'} - \set{x}\big)  \cup \big(\evn{\mat{Q}'} - \set{y} \big)  \big] \Big)
$
\\
$
\Omega= \varphi_1[c_z/c_x] \; \myet \; \varphi_2[c_z/c_y]
$
\\
$
\Gamma=\lambda_1[c_z/c_x] \;\cup\; \lambda_2[c_z/c_y]
$

\item \textbf{Hide operator}:
$$
\hspace{-4mm}
(5.a)\;
\frac
{\mat{P} \stackrel{\langle   _{E}a_x,\varphi,\lambda \rangle} {\fleche} \mat{P}' \;\;\;\; a \notin L}
{\mat{P} \setminus L \stackrel{\langle _{E}a_x,\varphi,\lambda \rangle} {\fleche} \mat{P}' \setminus L}
\hspace{3mm}
(5.b)\;
\frac
{\mat{P} \stackrel{\langle _{E}a_x,\varphi,\lambda \rangle} {\fleche} \mat{P}' \;\;\;\; a \in L}
{\mat{P} \setminus L \stackrel{\langle _{E}i_x,\varphi,\lambda \rangle} {\fleche} \mat{P}' \setminus L}
$$

\item \textbf{Interruption operator}:

 $$ 
(6.a) \;
\frac
 {\mat{P} \stackrel{\langle _{E}a_x, \varphi, \lambda \rangle} {\fleche} \mat{P}' \;\;\;\; a \neq \delta}
 {\mat{P} \inter \mat{Q} \stackrel{\langle _{E}a_y, \varphi[c_y/c_x], \lambda[c_y/c_x] \rangle} {\fleche}  \mat{P}' [y / x] \inter \mat{Q}} 
$$
$\hspace{22mm}y=get(\evns - (( \evn{\mat{P}'} - \set{x} ) \cup \evn{\mat{Q} } ))
$

$$
(6.b)\;
\frac
 {\mat{P} \stackrel{\langle _{E}\delta_x, \varphi, \lambda \rangle} {\fleche} \mat{P}' }
 {\mat{P} \inter \mat{Q} \stackrel{\langle _{E}\delta_x, \varphi, \lambda\rangle} {\fleche}  \mat{P}'  }
\;\;\;\;
  (6.c)\;
\frac
 {\mat{Q} \stackrel{\langle _{E}a_x, \varphi,\lambda \rangle } {\fleche} \mat{Q}' }
 {\mat{P} \inter \mat{Q} \stackrel{ \langle _{E}a_x, \varphi, \lambda \rangle } {\fleche}  \mat{Q}'  } 
$$

\item \textbf{Delay operator}:
$$
\frac
{\mat{P} \stackrel{\langle {}_{E}a_x , \; \varphi ,\; \lambda \rangle}{{\longrightarrow}} \mat{P}'} 
{\delay{d}\; \mat{P} \stackrel{\langle {}_E a_x, \;\varphi + d ,\; \lambda \rangle}{\longrightarrow} \mat{P}'}
$$
\end{enumerate}

The substitutions $\varphi[c_z/c_x]$ and $\lambda[c_z/c_x]$ as well as the union $\lambda_1  \cup \lambda_2$  are defined  in the most obvious way.  Now we define the function $\mat{F}^{\le u}$. Intuitively, the timed constraint $\fin{\le u}{E}$ of a given transition $t$ expresses that all the actions  in  $E$ must terminate and the transition $t$ can happen in the  time interval $[0,u]$ counting from the termination moment of  \emph{the last finished action(s) of $E$}, i.e. :  \label{F:eq:page}  
\begin{align}\label{F:eq}
\fin{\le u}{E}&= \myetb_{x:a \; \in E}\Big(\dr{a}\le c_x\Big)  \myet \myoub_{x:a \; \in E}\Big(c_x\le \dr{a}+u\Big) 
\end{align}
\begin{definition}\label{delay:fun:def}
The  delay  function $\varphi +  d$ is defined by induction on $\varphi$ as  follows:
\begin{align*}
 (\varphi_1\myet \varphi_2)+ d  &=(\varphi_1+ d) \myet (\varphi_2 + d)  \\
 (\varphi_1\myou \varphi_2)+ d  &=(\varphi_1+ d) \myou (\varphi_2 + d)  \\
 (\alpha \le c_x) + d&=   \alpha+d \le c_x \\
 ( c_x \le \beta)  + d &= c_x \le \beta +d  
\end{align*}
\end{definition}

\begin{remark}
By construction (i.e. by the construction of the timed constraints in the  rules $(1.a)$, $(1.b)$, $(2.a)$, $(2.b)$, $(4.c)$, and $7$), the timed constraints have the following form: 
\begin{align*}
\varphi &= \phi_1 \myet \cdots \myet \phi_n \\
\phi_i &= \myetb_{x:a \; \in E}\big(\alpha  \le c_x\big)  \myet \myoub_{x:a \; \in E}\big(  c_x\le \beta) \;\;\; \textrm{ where }\\ & \hspace{34mm} \alpha,\beta \in \dom \tand \alpha \le \beta.
\end{align*}
\end{remark}

We state one of the most properties of the function $\fin{\le u}{ . } + d$:
\begin{lemma}\label{fin:lemma}
Let $s_1 \stackrel{\langle _E{b}_x, \; \fin{\le u }{E}+d, \; c_x \rangle }{\fleche} s_2$ be a timed transition of a given timed-CTS. The action $b$ is enabled in the timed interval $[\tau+d,\tau+d+u]$ where $\tau \in \dom$ is the time stamp of the  termination of the last finished action(s) in $E$.   
\end{lemma}

\begin{definition}\label{densem:def}
Given a duration-CSP process $P$, the denotational semantics of $P$ over the class of timed-CTS, denoted by $\sem{P}$, consists in associating to $P$ the  timed-CTS  which  is generated by the  transition relation $\fleche \; \in \st  \times Act \times \st$  given in Section \ref{Den:sec},  starting from  the configuration  $_{\emptyset}[P]$.
 \end{definition}

\paragraph{Equivalence of the operational and denotational semantics} ~\\
 We arrive at the final point of this section: we  prove that the two semantics   are equivalent. The notion of equivalence is  formalized through the notion of $\tau$-bisimulation.  \\ 
Let $f: A \rTo B$ and let $A' \subseteq A$ and $B' \subseteq B$. The parametrized restrictions of $f$ w.r.t.  its domain and co-domain are defined respectively as follows:
\begin{align*}  
&  f_{\pi_1(A')} := \set{(a,b) \;|\; a \in A'} \hspace{1mm} \textrm{ and }
&   f_{\pi_2(B')} := \set{(a,b) \;|\;b \in B'}
\end{align*}

\begin{definition}\label{hsim:def}
A $\tau$-bisimulation linking the states of a timed-CTS and the timed causal configurations of $\tconf$ is a binary relation $\myequiv$ that comes with an events' bijection $f: \evns \flechep \evns$, and satisfying the following conditions:
\begin{enumerate}[{1}.1.] 
\item if $\aconf{s}{\nu} \stackrel{_Ea_x} {\fleche} \aconf{s'}{\nu'}$ then there exists $\mat{P}_{\tau} \stackrel{_{F_{\tau}}a_y} {\flecheop} \mat{P}_{\tau}'$ such that 
\begin{enumerate}[{i.}]
\item  $z:b \in E$ if and only if  $f(z):b:t \in F_{\tau}$,  for some $t\in \dom$, and
\item $(\aconf{s'}{\nu'},\mat{P}_{\tau}')_{f'} \in  \myequiv $ where \\ $\textrm{ }\hspace{12mm}f':= (f_{\pi_1(\psi(s')-x )})_{\pi_2(\psi(\mat{P}_{\tau}')-y )} \; \cup \set{(x,y)}$.
\end{enumerate}
\item if $\aconf{s}{\nu} \stackrel{d} {\fleche} \aconf{s}{\nu'}$ then $\mat{P}_{\tau} \stackrel{d} {\flecheop} \mat{P}_{\tau}'$ and $(\aconf{s}{\nu'},\mat{P}_{\tau}')_{f} \in \myequiv$.
 \end{enumerate}
\begin{enumerate}[{2}.1.] 
\item if $\mat{P}_{\tau} \stackrel{ _{F_{\tau}} a_y} {\flecheop} \mat{P}_{\tau}'$ then there exists $\aconf{s}{\nu} \stackrel{_Ea_x} {\fleche} \aconf{s'}{\nu'}$ such that 
\begin{enumerate}[{i.}]
\item  $z:b \in E $ if and only if  $f(z):b:t \in F_{\tau}$, for some $t\in \dom$, and
\item $(\aconf{s'}{\nu'},\mat{P}_{\tau}')_{f'} \in  \myequiv $ where \\ $\textrm{ }\; \hspace{12mm}f':= (f_{\pi_1(\psi(s')-x )})_{\pi_2(\psi(\mat{P}_{\tau}')-y )} \;\cup \set{(x,y)}$.
\end{enumerate}
\item if $\mat{P}_{\tau} \stackrel{d} {\flecheop} \mat{P}_{\tau}'$ then $\aconf{s}{\nu} \stackrel{d} {\fleche} \aconf{s}{\nu'}$ and $(\aconf{s}{\nu'}, \mat{P}_{\tau}')_{f'} \in \; \myequiv$.
 \end{enumerate}
A timed-CTS and a set of timed causal configuration are $\tau$-bisimilar iff there exists a $\tau$-bisimulation containing their initial configurations.
\end{definition}

\begin{theorem}\label{main:theorem}
The operational and the denotational semantics $( . )^{op} $ and $\sem{ . }$  are equivalent, i.e. for each duration-CSP process $P$ there exists a $\tau$-bisimulation $\myequiv$  such that  $(\sem{P},\;  P^{op})   \in  \myequiv $.
\end{theorem}

\section{Simple case study}
As a simple application we illustrate the use of duration-CSP through a simplified version of the Tick-Tock protocol \cite{tick-tock}, the latter has been used for the assessment of timed formal description techniques. \\
The tick-Tock case contains three entities called \emph{sender, receiver} and \emph{service}, see Figure \ref{sketch:fig}. Moreover, \emph{service} interacts with \emph{sender} and \emph{receiver} through their \texttt{SAPs} \texttt{Ss-SAP} and \texttt{Sr-SAP}, respectively. In the sequel we restrict ourselves to the specification of the \emph{service}. The description of the \emph{service} is as follows. \emph{service} transmits data from \emph{sender} to receiver. The exchanges are performed thought the corresponding \texttt{SAPs} in an atomic way and carried out a data called the \emph{cell}. \emph{Service} must satisfies the following requirements:\\
\begin{figure}[h]
\centering
\includegraphics[width=0.36\textwidth,height=0.26\textwidth]{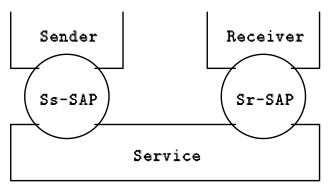}
\caption{The protocol.}
\label{sketch:fig}
\end{figure}

\textbf{Frequency.} A cell form \emph{sender} is only accepted from service at precise, punctual instants within  a period of $\pi$ units of time.\\
\textbf{Transmission delay.} \emph{Service} provides a cell to \emph{receiver} between $\tau_{min}$ and $\tau_{max}$ units of time after its emission.  \\
\textbf{Spacing between deliveries.} There is a delay of at least $\delta$ units of times between two consecutive offers of cells at \texttt{Sr-SAP}.\\
\textbf{Immediate acceptance.} A cell offered by \emph{service} to \emph{receiver} must be immediately accepted by \emph{receiver}, otherwise the service loses the cell immediately.\\
\textbf{Loss free transmission.} No cell is lost during its transmission through \emph{service}.
\subsection{Specification of \emph{service} with duration-CSP}
The specification of \emph{service} is given in such a way each timed requirement is given as a duration-CSP process. 

It is composed of three processes: \texttt{Frequency, Medium} and \texttt{ImmAccept}. \\
\textbf{Frequency.} The frequency behaviour of \emph{service} is:\\
\texttt{process Frequency[Ss-SAP]:=} \\
 \hspace{2mm}\texttt{Ss-SAP\{0\}; $\delay{\pi}$Frequency[Ss-SAP] + $\delay{\pi}$Frequency[Ss-SAP] endproc }\\

\noindent \textbf{Medium}. The \texttt{Medium} must satisfy both the transmission delay and spacing between deliveries requirements :\\
\texttt{process Medium[Ss-SAP,Del] := }\\
 \texttt{(Ss-SAP; TRANS; Del; Stop ||| Medium[Ss-SAP,Del] )} \\
 \hspace{3mm} \texttt {|[Del]|} \\
 \texttt{ Del; $\delay{\delta}$ Medium [Ss-SAP,Del] }\\
\texttt{endproc}\\

\noindent\textbf{Immediate acceptance.}  This requirement is specified as follows:\\
\texttt{process ImmAccept[Del,Sr-SAP]:= Del;}\\
\texttt{( Sr-SAP\{0\}; ImmAccept[Del,Sr-SAP])  + \\
  \hspace{3mm} ImmAccept[Del,Sr-SAP]}  \texttt{endproc}\\

\noindent\textbf{Service.} The three above processes have to synchronize on the internal action \texttt{Del}. Since  \texttt{Del} is an internal action, it must be hidden. The behaviour of teh process \texttt{Service} is as follows:\\
\begin{verbatim} process Service[Ss-Sap]:=
  (Frequency[Ss-SAP]  |[Ss-SAP]|
     ( Medium[Ss-SAP,Del] |[Del]| ImmAccept[Del,Sr-SAP]) 
  ) \{Del} 
  endproc
  \end{verbatim}

We note that all the actions are atomic apart the action \texttt{TRANS} we denotes the transmission delay. Therefore the duration of \texttt{TRANS} should belong to the interval $[\tau_{min},\tau_{max}]$.  As a matter of fact it is not hard to change the semantics of language by considering the actions to be of a variable duration instead of a fixed one. Finally we point out that one of the  interesting features of  duration-CSP - with its timed causal semantics- is that it allows the refinement of a given action, notably the action \texttt{TRANS} in this example,  into a more complicated process which allows an incremental design of the system. The refinement operator as well as its semantics and properties are discussed in the following section.

\section{Action refinement in duration-CSP}\label{ref:sec}
One of the interesting steps during the hierarchical design of complex systems is the refinement of an action $a$ into a process. As a matter of fact,   one can   associate to each specification a  level of abstraction basing on  the details of the actions with compose the specification.    For instance, given a  specification $E$ of abstraction level $N$,  the refinement $\reff{a}{P}{E}$ of an action $a$ by a process $P$ in the specification $E$ means that when passing from the abstraction level $N$ to $N+1$ the refinement operator will exhibits the internal structure of the action $a$,  that is,  $a$ would be replaced by the process $P$ at the level $N+1$.  There have been many earlier works  to curry on action refinement in process algebra, let us mention \cite{CourtiatS93,Saidouni94syntacticaction,FecherMW02,e-lotos:ref}.\\
In this section we enrich the language duration-CSP with the  refinement operator $\rho$. The new language is called $duration-CSP_{\rho}$, afterwards, we give the timed causal semantics of this language, notably, the semantics of the refinement operator. Finally we prove that the refinement operator preserves the timed causal bisimulation.  \\
 \textbf{The syntax of duration-CSP$_{\rho}$ } is given as follows:
\begin{itemize}
\item if $P$ is a duration-CSP process then $P$ is again a duration-CSP$_{\rho}$ process,
\item if $a$ is an action, $P$ is a duration-CSP process and $Q$ is a duration-CSP$_\rho$  process, then $\reff{a}{P}{Q}$ is a duration-CSP$_{\rho}$ process.
\end{itemize}
In order to define the timed causal semantics of the refinement operator $\rho$, we introduce a new kind of operator on the timed causal configurations $\tconf$,  called \emph{partial sequencing operator} and denoted by $\gg^{x}$.  Intuitively, the semantics of $\mat{P}_{\tau} \gg^{x} \mat{Q}_{\tau}$  means that all the actions of $\mat{Q}_{\tau}$ which do not depend on the termination  of the event $x$ are in concurrence with the actions of $\mat{P}_{\tau}$, however the execution of the remaining actions of $\mat{Q}_{\tau}$   must wait for  the successful termination of $\mat{P}_{\tau}$.  Besides the distributivity of the event names over the basic duration-CSP operators, we assume that the event names distribute over the refinement operator, i.e. for every $E_{\tau} \in 2^{\evns \times \mat{L}\times \mathbb{R}^{+}}$ and every process $\reff{a}{P}{Q}$,
$$
{}_{E_{\tau}}[\reff{a}{P}{Q}] \equiv \reff{a}{P}{{}_{E_{\tau}}[Q]}
$$  
Again we can extend Lemma \ref{canonical:lemma} to obtain:
\begin{lemma}\label{canonical:lemma:rho}
Every canonical  timed causal configuration  has one of the following forms:
\begin{align*}
&  _{E_{\tau}}[stop]   \hspace{3.5mm} _{E_{\tau}}[skip\set{d}]  \hspace{3.5mm} \delay{d} \mat{P}_{\tau}  \hspace{3.5mm} _{E_{\tau}}[a\set{d};P] \hspace{3.5mm} \mat{P}_{\tau} \ou \mat{Q}_{\tau}\\ &   \mat{P}_{\tau}|[L]|\mat{Q}_{\tau}  \hspace{3.5mm} \mat{P}_{\tau} \setminus L  \hspace{4.3mm} \mat{P}_{\tau} \inter \mat{Q}_{\tau}  \hspace{4.3mm} \reff{a}{P}{\mat{Q}_{\tau}} \hspace{4mm} \mat{P}_{\tau} \gg^{x} \mat{Q}_{\tau} 
\end{align*}
where $\mat{P}_{\tau}$ and $\mat{Q}_{\tau}$ are in the canonical  form. 
\end{lemma}
The function $\psi : \tconf \flechep 2^{\evns \times \actset \times \dom }$ that determines the set of events of a given timed configuration of duration-CSP$_{\rho}$ is the same as that of Definition  \ref{encap:def}  extended with the following rules:
\begin{align*}
 \evn{\mat{P} \gg^{x} \mat{Q}} &= \evn{\mat{P}} \cup (\evn{\mat{Q}} - \set{x}) \\  \evn{\reff{a}{P}{\mat{Q}}}&= \evn{\mat{Q}}
\end{align*}

\subsection{Operational semantics of duration-CSP$_{\rho}$}
This subsection introduce the operational semantics of duration-CSP$_{\rho}$ in the same way as we have done with duration-CSP.


\begin{definition}
The timed transition over the timed causal configurations of duration-CSP$_{\rho}$, denoted again  by $\flecheop$ is the relation that satisfies the rules 0,$\cdots$,VIII extended with the following rules:
 $$\textrm{R.1}
\frac{\mat{P}  \stackrel {_{E_{\tau}}a_y}{\flecheop} \mat{P}'}
{ \mat{P} \gg^{x} \mat{Q} \stackrel {_{E_{\tau}}a_z}{\flecheop} \mat{P}'[z/y] \gg^{x} \mat{Q}  }
$$
$z=get(\evns - \psi((\mat{P}') -\set{y}) \cup (\psi(\mat{Q}) -\set{x}))$

 $$\textrm{R.2}
\frac{\mat{P}  \stackrel {_{E_{\tau}}\delta_y}{\flecheop} \mat{P}'}
{ \mat{P} \gg^{x} \mat{Q} \stackrel {_{E_{\tau}}i_z}{\flecheop}  \mat{Q}[z/x]} \;\; z=get(\evns - (( \psi(\mat{Q}) - \set{x}))
$$
 $$\textrm{R.3}
\frac{\mat{Q}  \stackrel {_{E_{\tau}}a_y}{\flecheop} \mat{Q}' \;\; x \notin E_{\tau}}
{ \mat{P} \gg^{x} \mat{Q} \stackrel {_{E_{\tau}}a_z}{\flecheop} \mat{P} \gg^{x} \mat{Q}'[z/y]  }
$$
$
z= get (\evns - ((\psi(\mat{P}) \cup (\psi(\mat{Q}')-\set{y}) \cup \set{x}))
$
 $$\textrm{R.4}
\frac{\mat{Q}  \stackrel {_{E_{\tau}}b_y}{\flecheop} \mat{Q}' \;\;\;\; b\neq a}
{  \reff{a}{P}{\mat{Q}} \stackrel {_{E_{\tau}}b_x}{\flecheop}  \reff{a}{P}{\mat{Q}'}}
$$

 $$\textrm{R.5}\;\;
\frac{\mat{Q}  \stackrel {_{E_{\tau}}a_x}{\flecheop} \mat{Q}' \;\;\;\;
   _{E_{\tau}}[P]  \stackrel {_{E_{\tau}}b_y}{\flecheop} \mat{P}'}
{  \reff{a}{P}{\mat{Q}} \stackrel {_{E_{\tau}}b_z}{\flecheop}  \mat{P}'[z/y] \gg^{x}\reff{a}{P}{\mat{Q}'}}
$$
$z=get(\evns -  ((\psi(\mat{P}') -\set{y}) \cup (\psi(\mat{Q}') - \set{x}))) $

 $$\textrm{R.}\tau.1 \;
\frac{\mat{P}  \stackrel {d}{\flecheop} \mat{P}' \;\; x \in \evn{\mat{Q}}}
{ \mat{P} \gg^{x} \mat{Q} \stackrel {d}{\flecheop} \mat{P}' \gg^{x} \mat{Q} }
\;\;\; \textrm{R.}\tau.2 \;
\frac{\mat{P}  \stackrel {d}{\flecheop} \mat{P}' \;\; \mat{Q}  \stackrel {d}{\flecheop} \mat{Q}' \;\;x \notin \evn{\mat{Q}}}
{ \mat{P} \gg^{x} \mat{Q} \stackrel {d}{\flecheop} \mat{P}' \gg^{x} \mat{Q}' }
$$
 $$\textrm{R.}\tau.3\;\;
\frac{\mat{Q}  \stackrel {d}{\flecheop} \mat{Q}'}
{  \reff{a}{P}{\mat{Q}} \stackrel {d}{\flecheop}  \reff{a}{P}{\mat{Q}'}}
$$
\end{definition}
The rules R.1, R.2, R.3, R.$\tau$.1 and R.$\tau$.2 define the semantics of the partial sequencing operator $\gg^{x}$. That is, the rule R.1 expresses the fact that the occurrence of any action in the configuration  $\mat{P}$ remains possible in the configuration $\mat{P} \gg^{x} \mat{Q}$; however the renaming of the event $y$ is necessary because $y$ may be the event of some action which is already running in the configuration $\mat{Q}$.  The rule R.2 expresses the case of the successful  termination of $\mat{P}$. Note that the event $x$ is renamed with $z$ which identifies the successful termination of $\mat{P}$. The rule R.3 expresses that the occurrence of all the actions of the configuration $\mat{Q}$ which do not depend on the termination of the event $x$ -- i.e. on the successful termination of the configuration $\mat{P}$ -- can be executed  in the configuration $\mat{P} \gg^{x} \mat{Q}$.  The rule R.$\tau$.1 shows that the time is allowed only to elapse in the left part of the configuration $\mat{P} \gg^{x} \mat{Q}$  whenever  $\mat{Q}$ is waiting for the termination of the event $x$. However the rule R.$\tau$.2 allows the elapse of time in both parts of the configuration $\mat{P} \gg^{x} \mat{Q}$ if $\mat{Q}$ is not waiting for the termination of $x$. \\
The rules R.4, R.5 and R.$\tau$.3  give the semantics of the refinement operator $\rho$. The rule R.4 shows the case when the configuration $\mat{Q}$ provides  an action $b$ which is not subject to the refinement; in this case the action $b$ remains possible in the configuration $\reff{a}{P}{\mat{Q}}$.  The rule R.5 expresses the case when  the configuration $\mat{Q}$ provides the action $a$ which has to  be     refined into the process $P$. Hence the execution of the action $a$ must be replaced by the execution  of the process $P$. Since the execution of $a$ depends on the termination of all  the events of $E_{\tau}$, then every  action of $P$ depends also on the termination of the same set of events. Moreover, it is clear that all the actions of $\mat{Q}'$ which depend on the termination of $a$ must also depend on the successful termination of $_{E_{\tau}}[P]$, however the remaining actions are executed in parallel with $_{E_{\tau}}[P]$.  This shows the usefulness of the partial sequencing operator $\gg^{x}$ in expressing the semantics of the refinement operator.  

The following Theorem shows   the main property of the refinement operator $\rho$; it expresses that the refinement operator preserves the timed causal bisimulation\footnote{Indeed we mean the timed causal bisimulation that links the timed configurations and which is  defined in a routine way, see  the appendix Definition \ref{merde}.  }.  
 \begin{theorem} \label{ref:th}
 For every timed configuration $\mat{P}, \mat{Q}$ of duration-CSP$_{\rho}$, for every  action $a$ and for every duration-CSP process $E$, if $\mat{P}\tcsbis \mat{Q}$ then $\reff{a}{E}{\mat{P}}\tcsbis \reff{a}{E}{\mat{Q}}$.  
 \end{theorem}

\section{Current and future works}\label{conc:sec}
 At the moment we are looking for a probabilistic extension of the timed causal transition systems in the following way:  rather than considering that the actions have a fixed duration,  it is more realistic to attribute to them  a probabilistic duration that follows a certain distribution, notably a normal (Gaussian)  distribution. Within this  model, many problems suggest themselves such as the model checking one.    This is an orthogonal formalism  w.r.t. the probabilistic timed automata  \cite{probata} where  the probabilities   are  attributed  to the  transitions rather than the actions.   

An other work consists in considering the model checking of the duration logics \cite{CalcDurD91,InterLogic:Russia} over the timed causal transition systems.   \\

Finally we emphasize that it is not  useful  to encode the timed-CTS model into the timed automata one since this implies the loss of  the notion of true concurrency and gives rise to a combinatorial explosion due  to the fact of splitting  each action into two events: the starting and the finishing one. The implementation of an environment that integrates the timed-CTS model, the duration-CSP language and the refinement operator $\rho$ should not provide any technical difficulties.        

\bibliographystyle{alpha}
\bibliography{TimedBiblio}

\newpage

\section*{Appendix:  proofs of  the statements}
 
\setcounter{section}{4}
\setcounter{lemma}{0}
\begin{lemma}\label{canonical:lemma:app}
Every canonical  timed causal configuration in  $\tconf$  has one of the following forms:
\begin{align*}
&  _{E_{\tau}}[stop]   \hspace{5mm} _{E_{\tau}}[skip\set{d}]  \hspace{5mm} \delay{d} \mat{P}_{\tau}  \hspace{5mm} _{E_{\tau}}[a\set{d};P] 
 \hspace{5mm} \\ 
&  \mat{P}_{\tau} \ou \mat{Q}_{\tau}  \hspace{5mm} \mat{P}_{\tau}|[L]|\mat{Q}_{\tau}  \hspace{5mm} \mat{P}_{\tau} \setminus L  \hspace{5mm} \mat{P}_{\tau} \inter \mat{Q}_{\tau} 
\end{align*}
where $\mat{P}_{\tau}$ and $\mat{Q}_{\tau}$ are in the canonical  form. 
\end{lemma}

\begin{proof}
We prove by induction that every timed configuration which is not under one of these forms can be reduced by distributing the set of events over the algebraic operators.  The proof of the same lemma but upon the \emph{untimed} configurations was given in \cite{mahboulthese}, however we adapt it to the timed configurations.  

If a given timed configuration $\mat{R}_{\tau}'$ can  be obtained from $\mat{R}_{\tau}$ by distributing the set of events over the algebraic operators then we write $\mat{R}_{\tau} \reduc \mat{R}_{\tau}'$.  We only consider the cases where the timed configuration is of the form $_{E_{\tau}}[R]$:
\begin{itemize}
\item $R \cong \delay{d}P$:    $\quad  _{E_{\tau}}[R] \reduc  \delay{d} _{E_{\tau}}[P]$,
\item $R \cong P+Q$:  $\quad  _{E_{\tau}}[R] \reduc   \;  _{E_{\tau}}[P] + \; _{E_{\tau}}[Q]$.
\item $R \cong P|[L]|Q$: $\quad  _{E_{\tau}}[R] \reduc \;  _{E_{\tau}}[P] |[L]|  \; _{E_{\tau}}[Q]$, 
\item $R \cong P \setminus L$: $\quad  _{E_{\tau}}[R] \reduc    _{E_{\tau}}[P] \setminus L $,
\item $R \cong P \inter Q$: $\quad  _{E_{\tau}}[R] \reduc  \;  _{E_{\tau}}[P] \inter  \; _{E_{\tau}}[Q]$.
\end{itemize}
This ends the proof of Lemma \ref{canonical:lemma}.
\end{proof}

\setcounter{section}{5}
\setcounter{proposition}{2}
\begin{lemma}
Let $s_1 \stackrel{\langle _E{b}_x, \; \fin{\le u }{E}+d, \; c_x \rangle }{\fleche} s_2$ be a timed transition of a given timed-CTS. The action $b$ is enabled in the timed interval $[\tau+d,\tau+d+u]$ where $t\in \dom$ is the time stamp of the  termination of the last finished action in $E$.   
\end{lemma}
\begin{proof}
  Recall first the definition of $\mat{F}^{\le u}$ (see Equation (\ref{F:eq}) at page  \pageref{F:eq:page}): 
\begin{align}
\fin{\le u}{E}&= \myetb_{x:a \; \in E}\Big(\dr{a}\le c_x\Big)  \myet \myoub_{x:a \; \in E}\Big(c_x\le \dr{a}+u\Big) 
\end{align}
therefore by the definition of $+$ (see Definition \ref{delay:fun:def}), we get 

\begin{align*}
\fint{\le u}{d}{E}& =\\ &\hspace{-8mm} \underbrace{\myetb_{x:a \; \in E}\Big(\dr{a}+d \le c_x\Big)}_{\Phi_1}  \myet \underbrace{\myoub_{x:a \; \in E}\Big(c_x\le \dr{a}+ d + u \Big)}_{\Phi_2} 
\end{align*}

On the one hand, the constraint $\Phi_1$ ensures that the action $b$ is enabled in the interval $[\tau+d , \infty]$, where $\tau$ is the time stamp of termination of the last finished action in $E$. \\
On the other hand,  the condition $\Phi_2$ states that the  action $b$ is enabled in the interval $[0,\tau_{max}]$ where 
\begin{align*}
\tau_{max} \;= \; Max_{ x:a \in E} \; \set{\tau_x + d+ u} = Max_{ x:a \in E} \; \set{\tau_x} +  d + u
\end{align*}
 where  $\tau_x$ is the time stamp of the termination of the action $a$ s.t.  $x:a \in E$. Hence,
\begin{align*}
\tau_{max} \;= \; \tau  + d + u 
\end{align*} 
Therefore, the constraint $\Phi_2$ states that the action $b$ is enabled in the interval $[0,\tau + d + u]$.  We conclude that the constraint $\Phi_1 \myet \Phi_2$ states that the action $b$ is enabled in the interval $[\tau+d ,\tau+d+u]$.
\end{proof}

\setcounter{section}{5}
\setcounter{proposition}{5}
\setcounter{theorem}{0}
\begin{theorem}
The operational and the denotational semantics $( . )^{op} $ and $\sem{ . }$  are equivalent, i.e. for each duration-CSP process $P$ there exists a $\tau$-bisimulation $\myequiv$  such that $(\sem{P},\;  P^{op})   \in  \myequiv $.
\end{theorem}
\begin{proof}
We construct a binary relation $\myequiv$ linking the elements of $\sem{P}$ and $ P^{op}$ , afterward  we prove that it is a $\tau$-bisimulation. First of all we came assume that $\myequiv$ comes with the identity function  $Id: appendix.tex,v 1.20 2009/10/17 19:03:59 belkhir Exp $ over the set of events, i.e. we do not need to rename the events. \\
We let 
\begin{align*}
\myequiv=(\myequiv_0 \cup \hat{\myequiv}_0)\cup \cdots \cup (\myequiv_n \cup \hat{\myequiv}_n) \cup \cdots
\end{align*}
where 
\begin{align*}
  \myequiv_0=  & \set{(\aconf{{_{\emptyset}}[P]}{\zero}, {_{\emptyset}}[P]) } \cup  \\   \textrm{ } & \set{(\mat{P},\mat{Q}_{\tau}) \tst \exists d\in \mathbb{R} \; 
      \aconf{{_{\emptyset}}[P]}{\zero} \stackrel{d}{\fleche} \mat{P}  \\ & \tand {_{\emptyset}}[P]  \stackrel{d}{\flecheop} \mat{Q}_{\tau}) } \\
 \myequiv_{m+1}= & \set{(\mat{P}^{m+1}, \mat{Q}_{\tau}^{m+1})\; \tst \exists (\mat{P}^{m}, \mat{Q}_{\tau}^{m})\in \myequiv_{m} \tst \\ 
&  \mat{P}^{m} \stackrel{_{E}a_{x}} {\fleche} \mat{P}^{m+1} \tand \\
 &  \mat{Q}_{\tau}^{m} \stackrel{_{E_{\tau}}a_{x}}{\flecheop} \mat{Q}_{\tau}^{m+1} \textrm{ for some action } a  }\;\; \cup\\
& \set{(\mat{P}',\mat{Q}') \; \tst\;  \exists d \in \mathbb{R} \tst  \mat{P}^{m+1}\stackrel{d} {\fleche} \mat{P}' \\
& \;\;\; \tand   \mat{Q}_{\tau}^{m+1}\stackrel{d} {\flecheop} \mat{Q}' }\\
\end{align*}
 \begin{align*}
  \hat{\myequiv}_m& =\set{(\mat{P}^m, \bullet) \tst  (\mat{P}^m, \mat{Q}_{\tau}^{m})  \in \myequiv_m \\ & \;\;\; \tand \exists d \in \mathbb{R} \tst  {\mat{P}}^{m}\stackrel{d} {\fleche} \mat{P}' \tand   \mat{Q}_{\tau}^{m}\stackrel{d} {\not\flecheop} } \;\; \cup \\
& \;\;\;  \set{(\bullet,\mat{Q}_{\tau}^m) \tst  (\mat{P}^m, \mat{Q}_{\tau}^{m})  \in \myequiv_m \\ & \;\;\; \tand \exists d \in \mathbb{R} \tst  \mat{Q}_{\tau}^{m}\stackrel{d} {\flecheop} \mat{Q}' \tand   \mat{P}^{m}\stackrel{d} {\not\fleche} }
 \end{align*}
During the construction of $\myequiv_i, i=0,\cdots,n$, we require that the  invariants (\ref{INV1}) and (\ref{INV2})   hold.\\ 
The invariant (\ref{INV1}) is defined as follows: for each pair $(\aconf{\mat{R}}{\nu}, \mat{R}_{\tau}  ) \in \myequiv_n$, the pair $(\evn{\mat{R}_{\tau}}, \evn{\mat{R}})$ is \emph{synchronized} in the following sense:
\begin{align*}
z:b:t_{z} \in \evn{\mat{R}_{\tau}} \;\;  \textrm{ iff } \;\; z:b  \in \evn{\mat{R}} \tand t_{z} = \nu(c_z) \label{INV1} \tag{SYNCH1} 
\end{align*}
To give the definition of the  invariant \ref{INV2} we need some notations. Let us define the function $\for{.}$ that takes a timed configuration (in $\tconf$ or in  $\mathbb{C}$) and returns only the  duration-CSP process by deleting recursively the set of events:
\begin{align*}
  \for{_{E_{\tau}}[stop]}  &=   stop   \\  \for{ _{E_{\tau}}[skip\set{d}]}  &= skip\set{d} \\ 
 \for{_{E_{\tau}}[a\set{d};P]}&= a \set{d};P    \\ \for{\delay{d} \mat{P}_{\tau}} &= \delay{d} \for{\mat{P}_{\tau}} \\ 
 \for{\mat{P}_{\tau} \setminus L}&= \for{\mat{P}_{\tau}} \setminus L \\  \for{\mat{P}_{\tau} \ou \mat{Q}_{\tau}}&=  \for{\mat{P}_{\tau}} \ou \for{\mat{Q}_{\tau}} 
\\ \for{\mat{P}_{\tau}|[L]|\mat{Q}_{\tau}} &= \for{\mat{P}_{\tau}} |[L]| \for{\mat{Q}_{\tau}} \\ \for{\mat{P}_{\tau} \inter \mat{Q}_{\tau}}& = \for{\mat{P}_{\tau}} \inter \for{\mat{Q}_{\tau}}  
\end{align*} 
We let also, for $i\in \mathbb{N}$,  $\lbul{\myequiv_i}$  to be:
\[
\lbul{\myequiv_i} = 
  \begin{cases} 
   \set{(\aconf{_{\emptyset}[P]}{\zero},_{\emptyset}[P] ) } & \hspace{-15mm}\text{if } i=0 \\
    \set{({\mat{R}}^{i},{\mat{R}_{\tau}}^{i}) \in \myequiv_{i}  \tst \exists  ({\mat{R}}^{i-1},{\mat{R}_{\tau}}^{i-1}) \in \myequiv_{i-1} \tst \\ 
     \;\; \mat{R}^{i-1} \stackrel{_{E}a_x} {\fleche} \mat{R}^{i} \tand {\mat{R}_{\tau}}^{i-1} \stackrel{_{E_{\tau}}a_x} {\fleche} {\mat{R}_{\tau}}^{i} \textrm{ for some } a}    & \\ & \hspace{-15mm}\text{if } i\ge 1 
  \end{cases}
\]
The invariant (\ref{INV2}) is given by :
\begin{align*}
&\forall i \in \mathbb{N}, \forall (\aconf{\mat{R}^{i}}{\nu} , \mat{R}_{\tau}^{i}) \in \lbul{\myequiv_{i}} \textrm{ we have that } \\ 
& (i) \; \for{\mat{R}^{i}}= \for{\mat{R}_{\tau}^{i}} \textrm{ and } (ii) \textrm{ the pair} (\mat{R}^{i},\mat{R}_{\tau}^{i}) \textrm{ is synchronized}
\label{INV2} \tag{SYNCH2}
\end{align*}

Now we shall prove that $\myequiv$ is a $\tau$-bisimulation. For this aim, it is enough to prove that, for each $n \in \mathbb{N}$,  $\myequiv_n \cup \hat{\myequiv}_n$ is a $\tau$-bisimulation i.e.  $\hat{\myequiv}_n=\emptyset$.  The proof is by induction on $n$. \\

\noindent \textbf{Initial step} $n=0$. i.e. we consider $\myequiv_0$ defined by:
\begin{align*}
\myequiv_0=  & \set{(\aconf{{_{\emptyset}}[P]}{\zero}, {_{\emptyset}}[P]) } \cup  \\   \textrm{ } & \set{(\mat{P},\mat{Q}_{\tau}) \tst \exists d\in \mathbb{R} \; 
      \aconf{{_{\emptyset}}[P]}{\zero} \stackrel{d}{\fleche} \mat{P}  \\ & \tand {_{\emptyset}}[P]  \stackrel{d}{\fleche} \mat{Q}_{\tau}) }
 \end{align*}
In this step we shall prove that \emph{(i)} $\hat{\myequiv}_0=\emptyset$, \emph{(ii)} $\myequiv_0$ satisfies  the invariants (\ref{INV1}) and (\ref{INV2}) and \emph{(iii)} $\myequiv_{1}$ satisfies  the invariant (\ref{INV2}). The proof now is by structural induction on  $P$.
\begin{proofbycases}
\begin{caseinproof}The case $P=stop$ is obvious.
\end{caseinproof}
\begin{caseinproof}$P=skip\set{u}$.  The rule (1.$a$) of the denotational semantics ensures that $\forall d \; 0\le d \le u $ there is a  derivation 
\begin{align*}
\aconf{_{\emptyset}[skip\set{u}]}{\zero} \stackrel{d}{\fleche} \aconf{_{\emptyset}[skip\set{u}]}{\zero+d}
\end{align*}
 In the same way, the rule (I.$\tau$) of the operational semantics allows,   for each  $d \in ]0,u]$, the derivation  
\begin{align*}
_{\emptyset}[skip\set{u}] \stackrel{d}{\flecheop} \; _{\emptyset}[skip \set{u-d}]
\end{align*} 
This shows that $\hat{\myequiv}_0=\emptyset$, therefore $\myequiv_0 \cup \hat{\myequiv}_0$ is a $\tau$-bisimulation.
Now we show that  $\myequiv_0$  satisfies the invariants (\ref{INV1}), (\ref{INV2})  and  $\myequiv_1$ satisfies the invariant (\ref{INV2}).\\
Note that $\myequiv_0$ satisfies trivially the invariant (\ref{INV1})
because $\psi(_{\emptyset}[skip \set{u}])=
\psi(_{\emptyset}[skip\set{u-d}])=\emptyset$.  Also, $\myequiv_0$
satisfies trivially the invariant (\ref{INV2}) because
$\lbul{\myequiv_{0}}=\set{\big(\aconf{_{\emptyset}[skip\set{u}]}{\zero},
  \; _{\emptyset}[skip\set{u}] \big)}$. To show that $\myequiv_1$
satisfies the invariant (\ref{INV2}) we consider
$\lbul{\myequiv_1}$. The latter is obtained first  by applying the rule (1.$a$) of the denotational semantics  to $\aconf{_{\emptyset}[skip]}{\zero+d}$ giving arise to the
derivation:
\begin{align*}
\aconf{_{\emptyset}[skip]}{\zero+d} \stackrel{_{\emptyset}\delta_x  }{\fleche} \;\; \aconf{{}_{\set{x:\delta}}[stop]}{(\zero+d)[x \mapsto 0]}
 \end{align*}

And by applying the rule (I.a) of the operational semantics to the configuration $_{\emptyset}[skip\set{u-d}]$ giving arise to the derivation:
 \begin{align*}
{}_{\emptyset}[skip\set{u-d}] \stackrel{{}_{\emptyset}\delta_x}{\fleche} \;{}_{\set{x:\delta:0}}[stop]  
\end{align*}
Therefore 
\begin{align*}
\lbul{\myequiv_1}= \set{\big(\aconf{{}_{\set{x:\delta}}[stop]}{(\zero+d)[x \mapsto 0]},\; {}_{\set{x:\delta:0}}[stop]\big)}
\end{align*}
Note that $\myequiv$ satisfies the invariant (\ref{INV2}) because 
\begin{align*}
(i) \; \for{{}_{\set{x:\delta}}[stop]}= \for{{}_{\set{x:\delta:0}}[stop]}=stop
\end{align*}
and \\
$(ii)$ clearly the pair  $ \big(\aconf{{}_{\set{x:\delta}}[stop]}{(\zero+d)[c_x \mapsto 0]},\; {}_{\set{x:\delta:0}}[stop]\big)$ \\ is synchronized since the clock $c_x$ is reset to zero. 
 \end{caseinproof}

\begin{caseinproof} The case $P=a\set{u};Q$ is similar to the previous one apart that 
  we deal here with the action $a$ instead of $\delta$, and with the process $Q$ instead of the process $stop$. 
\end{caseinproof}

\begin{caseinproof} The case $P=Q\ou R$ is straightforward by applying the induction hypothesis to $Q$ and $R$. 
\end{caseinproof}

\begin{caseinproof}$P=P_1 |[L]| P_2$.  First we show  that $\hat{\myequiv}_0=\emptyset$.  The rule IV.$\tau$ of the operational semantics implies that  that if 
\begin{align*}
_{\emptyset}[\; P_1 |[L]| P_2\; ] \stackrel{d}{\flecheop}\; _{\emptyset}[P_1'] |[L] [P_2']
\end{align*}
then 
\begin{align*}
_{\emptyset}[P_i] \stackrel{d} {\flecheop}\;  _{\emptyset}[P_i']\;\;\;\;\; i=1,2.
\end{align*}
By applying the induction hypothesis to both $P_1$ and $P_2$ we get  the possible derivations:
\begin{align*}   
\aconf{_{\emptyset}[P_i]}{\zero} \stackrel{d} {\fleche}  \aconf{_{\emptyset}[P_i]}{\zero + d} \;\;\;\;\;\; i=1,2
\end{align*}
Hence 
\begin{align*}
\big(\; \aconf{_{\emptyset}[P_1 |[L]|  P_2]}{\zero} \stackrel{d} {\fleche}  \aconf{_{\emptyset}[P_1 |[L]|  P_2]}{\zero+d}\; \big)
\end{align*}
This shows that $\hat{\myequiv}_0=\emptyset$. Note that $\myequiv_{0}$ satisfies  the invariants (\ref{INV1}) and (\ref{INV2}) (the same arguments used in \textbf{{\emph{ Case (ii)}}} hold). Let us show that $\myequiv_{1}$ satisfies the invariant 
(\ref{INV2}).  To this goal let $a$ be an action, we consider  the case when $a\notin L\cup \set{\delta}$ and $i=1$. The case  when $a\notin L\cup\set{\delta}, \; i=2$ and the case when $a \in L \cup \delta$ are handled similarly.  Let $i=1$ and  assume the derivation:
\begin{align}\label{par:1}
_{\emptyset}[P_1'] \stackrel{{}_{\emptyset}a_x }{\flecheop} \;{}_{\set{x:a:0}}[Q_1']
\end{align}
The induction hypothesis shows that the following derivation is possible:
\begin{align}\label{par:2}
\aconf{_{\emptyset}[P_1]}{\zero+d} \stackrel{_{\emptyset}a_x}{\fleche} \aconf{_{\set{x:a}}[Q_1]}{\zero+d[c_x \mapsto 0]}
\end{align}
and ensures that $\for{_{\set{x:a:0}}[Q_1']}=\for{_{\set{x:a}}[Q_1]}= Q_1$.
Therefore by applying the rule (I.V.a) of the operational semantics and considering the derivation (\ref{par:1}) above we get the derivation:
\begin{align*}
_{\emptyset}[P_1']\; |[L]| \;_{\emptyset}[P_2] \stackrel{_{\emptyset}a_x}{\flecheop} \;_{\set{x:a:0}}[Q_1] \;|[L]|\; _{\emptyset}[P_2]   
\end{align*}
Also by applying the rule (4.a) of the denotational semantics and considering the rule (\ref{par:2}) above we get the derivation:
\begin{align*}
\aconf{_{\emptyset}[P_1 |[L]| P_2]}{\zero+d} \stackrel{_{\emptyset}a_x} {\fleche} \aconf{_{\set{x:a}}[Q_1] \;|[L]|\; _{\emptyset}[P_2]}{\zero +d [c_x \mapsto 0]}
\end{align*}
Thus
\begin{align*}
\lbul{\myequiv}_1=\set{\big( \aconf{_{\set{x:a}}[Q_1] \;|[L]|\; _{\emptyset}[P_2]}{\zero +d [c_x \mapsto 0]},  \\  _{\set{x:a:0}}[Q_1] \;|[L]|\; _{\emptyset}[P_2]\big)} 
\end{align*}
and it is easy to check that $\myequiv_1$ satisfies the invariant (\ref{INV2}).
\end{caseinproof}
\begin{caseinproof}
The cases of the hide operator (rules (V.a) and (V.b) ) and of the interruption operator (rules (VI.a), (VI.b) and (VI.c) ) are handled by the induction machinery.
\end{caseinproof}
\begin{caseinproof}
If $P=\delay{d} Q$, then it suffices to prove the following Claim:
\begin{claim}
Let $d\in \dom$,  $tr= s \stackrel{\langle _{E}a_{x}, \varphi, \lambda \rangle}{\fleche}s'$ be a  transition of a given time-CTS and $tr_{+}$ be the same transition apart that we replace $\varphi$ with $\varphi + d$, i.e.  $tr_{+}= s_{+} \stackrel{\langle _{E}a_{x}, \varphi +  {d}, \lambda \rangle}{\fleche}s_{+}'$. Then, the transition $tr$ allows  the action $a$ at the time stamp $\tau$ if and only $tr_{+}$ allows $a$ at the time stamp $\tau+d$.
\end{claim} 
\begin{proof}{[of the Claim]}
Straightforward from the definition of the delay function $+$ (see Definition \ref{delay:fun:def}) since $\varphi +  d$ lifts  every (atomic) constraint $\alpha \le c_x  $ to $\alpha +d \le c_x $,  and $c_{x} \le \beta$ to $c_{x} \le \beta+ d $. This ends the proof of the Claim.
\end{proof}
\end{caseinproof}

\end{proofbycases}

\textbf{Induction step:} $n>0$.\\
That is, we consider $ \myequiv_{n}$ defined above by :
\begin{align*}
 \myequiv_{n}= & \set{(\mat{P}, \mat{Q}_{\tau})\; \tst \exists (\mat{P}^{n-1}, \mat{Q}_{\tau}^{n-1})\in \myequiv_{n-1} \tst \\ 
&  \mat{P}^{n-1} \stackrel{_{E}a_{x}} {\fleche} \mat{P} \tand \\
 &  \mat{Q}_{\tau}^{n-1} \stackrel{_{E_{\tau}}a_{x}}{\flecheop} \mat{Q}_{\tau} \textrm{ for some action } a  }\;\; \cup\\
& \set{(\mat{P}',\mat{Q}') \; \tst\;  \exists d \in \mathbb{R} \tst  \mat{P}\stackrel{d} {\fleche} \mat{P}' \\
& \;\;\; \tand   \mat{Q}_{\tau}\stackrel{d} {\flecheop} \mat{Q}' }\\
\end{align*}
We recall that the induction hypothesis implies that $\myequiv_{n-1}$ satisfies the invariants (\ref{INV1}) and (\ref{INV2}), and that $\myequiv_n$ satisfies the invariant (\ref{INV2}). 
As we have done in the initial step, in this step we shall prove that \emph{(i)} $\hat{\myequiv}_n=\emptyset$, \emph{(ii)} $\myequiv_n$ satisfies  the invariants (\ref{INV1}) and (\ref{INV2}) and \emph{(iii)} $\myequiv_{n+1}$ satisfies  the invariant (\ref{INV2}). As a consequence of the induction hypothesis  $\myequiv_{n}$  may be written as:
\begin{align*}
 \myequiv_{n}= & \set{ (\aconf{_{E}[P]}{\nu} \; , _{E_{\tau}}[P]) \; \tst \exists (\mat{P}^{n-1}, \mat{Q}_{\tau}^{n-1})\in \myequiv_{n-1} \tst \\ 
&  \mat{P}^{n-1} \stackrel{_{E}a_{x}} {\fleche} \;  _{E}[P] \tand \\
 &  \mat{Q}_{\tau}^{n-1} \stackrel{_{E_{\tau}}a_{x}}{\flecheop}\;  _{E_{\tau}}[P] \textrm{ for some action } a  }\;\; \cup\\
& \set{(\mat{P}',\mat{Q}') \; \tst\;  \exists d \in \mathbb{R} \tst  \; \aconf{{_{E}}[P]}{\nu}\stackrel{d} {\fleche} \mat{P}' \\
& \;\;\; \tand   {_{E_{\tau}}}[P]\stackrel{d} {\flecheop} \mat{Q}' }\\
\end{align*}
where the pair $({E}, E_{\tau})$ is synchronized.
Again, the proof  is by structural induction on  $P$ and similar to the one given in the initial step.
\begin{proofbycases}
\begin{caseinproof}The case $P=stop$ is obvious because the pair 
$(E , E_{\tau})$ is synchronized.
\end{caseinproof}
\begin{caseinproof}$P=skip\set{u}$.  The rule (1.$b$) of the denotational semantics ensures that $\forall d \; 0\le d \le u $  and counting form the moment when all the actions of $E$ have finished (see the definition of $\mat{F}^{\le u}(.)$),  there is a  derivation 
\begin{align*}
\aconf{_{E}[skip\set{u}]}{\nu} \stackrel{d}{\fleche} \aconf{_{E}[skip\set{u}]}{\nu+d}
\end{align*}
 In the same way, the rule (I.$\tau$) of the operational semantics allows,   for each  $d \in ]0,u]$,  such that all the actions of $E_{\tau}$ have finished,  the derivation  
\begin{align*}
_{E_{\tau}}[skip\set{u}] \stackrel{d}{\flecheop} \; _{E_{\tau}}[skip \set{u-d}]
\end{align*} 
Since the pair $(E,E_{\tau})$  is synchronized thus  $\hat{\myequiv}_n=\emptyset$, therefore $\myequiv_n \cup \hat{\myequiv}_n$ is a $\tau$-bisimulation.
Using the same arguments of the \textbf{\emph{Case(ii) }} of the initial step
 one can  we show easily  that  $\myequiv_n$  satisfies the invariants (\ref{INV1}), (\ref{INV2})  and  $\myequiv_{n+1}$ satisfies the invariant (\ref{INV2}).\\
 \end{caseinproof}

\begin{caseinproof} The case $P=a\set{u};Q$ is similar to the previous one apart that 
  we deal here with the action $a$ instead of $\delta$, and with the process $Q$ instead of the process $stop$. 
\end{caseinproof}
\begin{caseinproof} The case $P=Q\ou R$ is straightforward by applying the induction hypothesis to $Q$ and $R$. 
\end{caseinproof}

\begin{caseinproof}$P=P_1 |[L]| P_2$.  First we show  that $\hat{\myequiv}_n=\emptyset$.  The rule (IV.$\tau$) of the operational semantics implies that  that if 
\begin{align*}
_{E_{\tau}}[\; P_1 |[L]| P_2\; ] \stackrel{d}{\flecheop}\; _{E_{\tau}}[P_1'] |[L] [P_2']
\end{align*}
then 
\begin{align*}
_{E_{\tau}}[P_i] \stackrel{d} {\flecheop} _{E_{\tau}}[P_i']\;\;\;\;\; i=1,2.
\end{align*}
By applying the induction hypothesis to both $P_1$ and $P_2$ we get  the possible derivations:
\begin{align*}   
\aconf{_{E_{\tau}}[P_i]}{\zero} \stackrel{d} {\fleche}  \aconf{_{E_{\tau}}[P_i]}{\zero + d} \;\;\;\;\;\; i=1,2
\end{align*}
Hence 
\begin{align*}
\big(\; \aconf{_{E}[P_1 |[L]|  P_2]}{\nu} \stackrel{d} {\fleche}  \aconf{_{E}[P_1 |[L]|  P_2]}{\nu+d}\; \big)
\end{align*}
Since the pair $(E,E_{\tau})$ is synchronized, then   $\hat{\myequiv}_n=\emptyset$. Note that for the same reason,  $\myequiv_{n}$ satisfies trivially the invariants (\ref{INV1}) and (\ref{INV2}). Let us show that $\myequiv_{1}$ satisfies the invariant 
(\ref{INV2}).  To this goal let $a$ be an action, we consider  the case when $a\notin L\cup \set{\delta}$ and $i=1$. The case  when $a\notin L\cup\set{\delta}, \; i=2$ and the case when $a \in L \cup \delta$ are handled similarly.  Let $i=1$ and  assume the derivation:
\begin{align}\label{par:11}
_{E_{\tau}}[P_1'] \stackrel{{}_{E_{\tau}}a_x }{\flecheop} \;{}_{\set{x:a:0}}[Q_1']
\end{align}
The induction hypothesis shows that the following derivation is possible:
\begin{align}\label{par:12}
\aconf{_{E}[P_1]}{\nu+d} \stackrel{_{E}a_x}{\fleche} \aconf{_{\set{x:a}}[Q_1]}{\nu+d[c_x \mapsto 0]}
\end{align}
and ensures that $\for{_{\set{x:a:0}}[Q_1']}=\for{_{\set{x:a}}[Q_1]}= Q_1$.
Therefore by applying the rule (I.V.a) of the operational semantics and considering the derivation (\ref{par:11}) above we get the derivation:
\begin{align*}
_{E_{\tau}}[P_1']\; |[L]| \;_{\emptyset}[P_2] \stackrel{_{E_{\tau}}a_x}{\flecheop} \;_{\set{x:a:0}}[Q_1] \;|[L]|\; _{\emptyset}[P_2]   
\end{align*}
Also by applying the rule (4.a) of the denotational semantics and considering the rule (\ref{par:12}) above we get the derivation:
\begin{align*}
\aconf{_{\emptyset}[P_1 |[L]| P_2]}{\zero+d} \stackrel{_{\emptyset}a_x} {\fleche} \aconf{_{\set{x:a}}[Q_1] \;|[L]|\; _{\emptyset}[P_2]}{\zero +d [c_x \mapsto 0]}
\end{align*}
Thus
\begin{align*}
\lbul{\myequiv}_1=\set{\big( \aconf{_{\set{x:a}}[Q_1] \;|[L]|\; _{\emptyset}[P_2]}{\zero +d [c_x \mapsto 0]},  \\  _{\set{x:a:0}}[Q_1] \;|[L]|\; _{\emptyset}[P_2]\big)} 
\end{align*}
and it is easy to check that $\myequiv_1$ satisfies the invariant (\ref{INV2}).
\end{caseinproof}
\begin{caseinproof}
The cases of the hide operator (rules (V.a) and (V.b) ),  of the interruption operator (rules (VI.a), (VI.b) and (VI.c) ), and the delay operator are handled by the induction machinery.
\end{caseinproof}

\end{proofbycases}
This ends the proof of Theorem \ref{main:theorem}. \qed
\end{proof}

\setcounter{theorem}{1}
 \begin{theorem} 
 For every timed configuration $\mat{P}, \mat{Q}$ of duration-CSP$_{\rho}$, for every  action $a$ and for every duration-CSP process $E$, if $\mat{P}\tcsbis \mat{Q}$ then $\reff{a}{E}{\mat{P}}\tcsbis \reff{a}{E}{\mat{Q}}$.  
 \end{theorem}
\begin{proof}
First we construct a binary relation linking the elements of $\reff{a}{E}{\mat{P}}$ and $\reff{a}{E}{\mat{Q}}$, and second we prove that it is a timed causal bisimulation. \\ 
We let $\myequiv=\myequiv_{1} \cup \myequiv_{2}$ where 
\begin{align*}
\myequiv_1= \set{(\reff{a}{E}{\mat{P}},\reff{a}{E}{\mat{Q}})_f\; \textrm{ s.t }\; (\mat{P},\mat{Q})_{f} \in \myequiv'} 
\end{align*}
such that $\myequiv'$ is a timed causal bisimulation, such bisimulation does exist by the hypothesis of the Theorem.

\begin{align*}
\myequiv_2=\set{(\mat{P}\gg^{x}\mat{P}^{+} , \mat{P}\gg^{y}\mat{Q}^{+} )_{f}\; \textrm{ s.t } (\mat{P}^{+}[v/x], \mat{Q}^{x}[w/y])_{f'} \in \myequiv}
\end{align*}
where 
\begin{align*}
& v \notin \big (\psi (\mat{P}^{+}) -\set{x}\big )  \; \cup f^{-1}\big(\psi(\mat{Q}^{+}) - \set{y}\big),  \\
& w \notin  f \big (\psi(\mat{P}^{+}) - \set{x} \big ) \cup  \big(\psi(\mat{Q}^{+}) - \set{y}\big),\\
& f'= f_{\pi_1(\psi({P})^{+} - \set{x} )} \cup f^{-1}\big(\psi(\mat{Q}^+ - \set{y}) \big), \tand \\
& f= f' \cup Id_{\pi_1(\psi(\mat{P}))}.
\end{align*}
Now we show that $\myequiv$ is a timed causal bisimulation.\\

\textbf{Initial step}\\
That is, we verify that $\myequiv_1$ is a timed causal bisimulation:
\begin{enumerate}
\item If $\reff{a}{E}{\mat{P}}\stackrel{_{E_{\tau}}b_{x}}{\fleche} \mat{H} $ then we distinguish two  cases according to $\mat{H}$:
\begin{itemize} 
\item $\mat{H}\equiv \reff{a}{E}{\mat{P}'}$, therefore $\mat{P}\stackrel{_{E_{\tau}}b_{x}}{\fleche} \mat{P}'$  and $a\neq b$. According to the hypothesis there exists a derivation $\reff{a}{E}{\mat{Q}}\stackrel{_{F_{\tau}}b_{y}}{\fleche} \reff{a}{E}{\mat{Q}'}$ such that 
\begin{enumerate}[(a)]
\item the definition  of $f$ ensures that for each $u \in \psi(\reff{a}{E_{\tau}}{\mat{P}})$, if $u \notin E$ and $f(u) \in \psi(\reff{a}{E}{Q})$ then $f(u) \notin F_{\tau}$, 
\item since  there exist $v,w \in \evns$ such that \\ $(\mat{P}'[v/x],\mat{Q}'[w/y])_{f''} \in \myequiv'$ where 
\begin{align*} f''= f_{\pi_1(\psi(\mat{P}') - \set{x})} & \cup f^{-1}(\psi(\mat{Q'})-\set{y}) \\ &\cup \set{(v,w)}
\end{align*}   
by using the definition  \ref{} it follows that 
\begin{align*}
(\reff{a}{E}{\mat{P}'[v/x]},\reff{a}{E}{\mat{Q}'[w/y]})_{\bar{f}} \in \myequiv_1
\end{align*}
where 
\begin{align*}\bar{f}=  f_{\pi_1(\psi(\mat{P}' - \set{x}))} &\cup f^{-1}(\psi(\mat{Q}')-\set{x}) \\ & \cup \set{(v,w)}
\end{align*}  
\end{enumerate}
\item $\mat{H}\equiv \mat{R} \gg^{z} \reff{a}{E}{\mat{P}'}$, then $\mat{P}\stackrel{_{E_{\tau}}a_{z}} {\fleche} \mat{P}'$ and $_{\emptyset}[E]\stackrel{_{\emptyset}b_y} {\fleche} \mat{R}$. According to the hypothesis we have that $\mat{Q}\stackrel{_{F_{\tau}}a_s}{\fleche} \mat{Q}'$, and by taking \\ $x \notin \psi(\mat{P}-\set{z}) \cup \psi(\mat{Q}') -\set{s}$, it follows that \\
$\reff{a}{E}{\mat{Q}} \stackrel{_{F_{\tau}}b_x} {\fleche} \mat{R} \gg^{s} \reff{a}{E}{\mat{Q}'}$
\end{itemize}
\item similar to 1.
\item if $\reff{a}{E}{\mat{P}} \stackrel{d}{\fleche}\reff{a}{E}{\mat{P}'} $, then $\mat{P} \stackrel{d}{\fleche}\mat{P}'$. According to the hypothesis there exists a derivation  $\mat{Q} \stackrel{d}{\fleche}\mat{Q}'$  such that $(\mat{P}',\mat{Q}')_{f} \in \myequiv_1$ for some $f$, therefore it follows that 
$(\reff{a}{E}{\mat{P}'},\reff{a}{E}{Q'})_{f} \in \myequiv_1$.
\end{enumerate}

\textbf{Induction step.}
In this step we consider the elements of $\myequiv_2$, these elements are of the form 
$(\mat{P}\gg^{x}\mat{P}^{+},\mat{P}\gg^{y}\mat{Q}^{+})_{f}$  where $(\mat{P}^{+}[v/x], \mat{Q}^{+}[w/y])_{f'} \in \myequiv$ with \\
$f'=f_{\pi_1(\psi(\mat{P})^{+} - \set{x})} \cup f^{-1}(\psi(\mat{Q}^{+}) - \set{y}) \cup \set{(u,w)}$:
\begin{enumerate}[{1}.1.]
\item $\mat{P}\gg^{x}\mat{P}^{+} \stackrel{_{E_{\tau}}a_z} {\fleche} \mat{H}$, we distinguish three cases according to $\mat{H}$:
\begin{itemize}
\item $\mat{H} \equiv \mat{P}' \gg \mat{P}^{+}$, then $\mat{P}\stackrel{_{E_{\tau}}a_z} {\fleche} \mat{P}'$.  By assuming that $z \notin \psi(\mat{P}' \gg^{x} \mat{P}^{+}) \cup f^{-1}\big(\psi(\mat{Q}^{+}) \cup \set{y}\big)  $, and applying the rule R.1  we obtain the derivation \\
$\mat{P} \gg^{y} \mat{Q}^{+}  \stackrel{_{E_{\tau}}a_z} {\fleche}\mat{P}' \gg^{y} \mat{Q}^{+}$, and we have done.    
\item $\mat{H} \equiv\mat{P}^{+}[z/x]$, then $\mat{P} \stackrel{_{E_{\tau}}\delta_z}{\fleche} \mat{P}'$ and $a=i$.    By assuming that $z \notin  \psi(\mat{P}^{+}) \cup f^{-1}\big( \psi(\mat{Q}^{+}) -\set{y} \big)$, and applying the rule R.2  we obtain the derivation \\
$\mat{P} \gg^{y} \mat{Q}^{+}  \stackrel{_{E_{\tau}}i_z} {\fleche} \mat{Q}^{+}[z/y]$, and we have done.
\item  $\mat{H} \equiv \mat{P} \gg^{x} {\mat{P}^{+}}'$,  then the rule R.3 implies that $\mat{P}^{+} \stackrel{_{E_{\tau}}a_z }{\fleche} {\mat{P}^{+}}'$; by applying the induction hypothesis there exists a derivation \\
$\mat{Q} \stackrel{_{F_{\tau}}a_s} {\fleche}{\mat{Q}^{+}}'$. By assuming $s \notin \psi(\mat{P}')$ we get \\
$
\mat{P} \gg^{y} \mat{Q}^{+} \stackrel{_{E_{\tau}}a_s} {\fleche} \mat{P} \gg^{y} {\mat{Q}^{+}}'.
$
\end{itemize}
\item $\mat{P} \gg^{x} \mat{P}^{+} \stackrel{d}{\fleche} \mat{H}$, we distinguish two cases according to $\mat{H}$:
\begin{itemize}
\item $\mat{H}\equiv \mat{P}' \gg^{x} \mat{P}^{+}$, the rule R.$\tau$.1 implies that \\ $\mat{P} \stackrel{d}{\fleche} \mat{P}' $ with $x \in \psi(\mat{P}^{+})$.  The induction hypothesis ensures that $(\mat{P}^{+}, \mat{Q}^{+}) \in \myequiv$.  Hence, $y \in \psi(\mat{Q}^{+})$.   By applying the rule R.$\tau$.1 we get the  derivation $\mat{P} \gg^{y} \mat{Q}^{+} \stackrel{d}{\fleche} \mat{P}' \gg^{y} \mat{Q}^{+} $.
\item $\mat{H} \equiv \mat{P}' \gg^{x} {\mat{P}^{+}}'$, the rule R.$\tau$.2 implies that  $\mat{P} \stackrel{d}{\fleche} \mat{P}' $   and   $\mat{P}^{+} \stackrel{d}{\fleche} {\mat{P}^{+}}' $ with $x \notin \psi(\mat{P}^{+})$. Hence $y \notin \psi(\mat{Q}^{+})$. By applying the rule R.$\tau$.2 we get the derivation:
 $\mat{P} \gg^{y} \mat{Q}^{+} \stackrel{d}{\fleche} \mat{P}' \gg^{y} {\mat{Q}^{+}}' $.
\end{itemize}
\end{enumerate}
\end{proof}

\section{On the timed causal bisimulation over the timed configurations}
\begin{definition}\label{merde}
A $\tau$-bisimulation linking  the timed causal configurations of $\tconf$ is a binary relation $\myequiv$ that comes with an events' bijection $f: \evns \flechep \evns$, and satisfying the following conditions:
\begin{enumerate}[{1}.1.] 
\item if $\mat{Q}_{\tau} \stackrel{_{E_{\tau}}a_x} {\flecheop} \mat{Q}_{\tau}'$ then there exists $\mat{P}_{\tau} \stackrel{_{F_{\tau}}a_y} {\flecheop} \mat{P}_{\tau}'$ such that 
\begin{enumerate}[{i.}]
\item  $z:b:t \in E_{\tau}$ if and only if  $f(z):b:t \in F_{\tau}$,  for some $t\in \dom$, and
\item $(\mat{Q}_{\tau}',\mat{P}_{\tau}')_{f'} \in  \myequiv $ where \\ $\textrm{ }\hspace{12mm}f':= (f_{\pi_1(\psi(\mat{Q}_{\tau}')-x )})_{\pi_2(\psi(\mat{P}_{\tau}')-y )} \; \cup \set{(x,y)}$.
\end{enumerate}
\item if $\mat{Q}_{\tau} \stackrel{d} {\flecheop} \mat{Q}_{\tau}' $ then $\mat{P}_{\tau} \stackrel{d} {\flecheop} \mat{P}_{\tau}'$ and $(\mat{Q}_{\tau}' ,\mat{P}_{\tau}')_{f} \in \myequiv$.
 \end{enumerate}
\begin{enumerate}[{2}.1.] 
\item if $\mat{P}_{\tau} \stackrel{ _{F_{\tau}} a_y} {\flecheop} \mat{P}_{\tau}'$ then there exists $\mat{Q}_{\tau} \stackrel{_{E_{\tau}}a_x} {\flecheop} \mat{Q}_{\tau}'$ such that 
\begin{enumerate}[{i.}]
\item  $z:b:t \in E_{\tau} $ if and only if  $f(z):b:t \in F_{\tau}$, for some $t\in \dom$, and
\item $(\mat{Q}_{\tau}',\mat{P}_{\tau}')_{f'} \in  \myequiv $ where \\ $\textrm{ }\; \hspace{12mm}f':= (f_{\pi_1(\psi(\mat{Q}_{\tau}')-x )})_{\pi_2(\psi(\mat{P}_{\tau}')-y )} \;\cup \set{(x,y)}$.
\end{enumerate}
\item if $\mat{P}_{\tau} \stackrel{d} {\flecheop} \mat{P}_{\tau}'$ then $\mat{Q}_{\tau} \stackrel{d}{\flecheop} \mat{Q}_{\tau}'$ and $(\mat{Q}_{\tau}', \mat{P}_{\tau}')_{f'} \in \; \myequiv$.
 \end{enumerate}

\end{definition}

\end{document}